\documentclass[11pt]{article}

\usepackage[T1]{fontenc}
\usepackage{lmodern}
\usepackage[letterpaper,margin=1in]{geometry}
\usepackage{amsmath,amssymb,amsthm}
\usepackage{mathtools}
\usepackage{array,tabularx,booktabs}
\usepackage{enumitem}
\usepackage{algorithm}
\usepackage{float}
\usepackage{algorithmic}
\usepackage{tikz}
\usetikzlibrary{arrows.meta}
\usepackage[section]{placeins}
\usepackage[numbers,sort&compress]{natbib}
\usepackage{aliascnt}
\usepackage{microtype}
\usepackage[colorlinks=true,linkcolor=blue,citecolor=blue,urlcolor=blue]{hyperref}
\usepackage[capitalise,nameinlink,noabbrev]{cleveref}

\hypersetup{
    pdftitle={GSP Autobidding Equilibria with Few Advertisers or Slots},
    pdfauthor={Yiyang Huang, Zhengyang Liu, Yonglei Yan, Zihe Wang}
}
\theoremstyle{plain}
\newtheorem{theorem}{Theorem}[section]
\newaliascnt{lemma}{theorem}
\newtheorem{lemma}[lemma]{Lemma}
\aliascntresetthe{lemma}
\newaliascnt{proposition}{theorem}
\newtheorem{proposition}[proposition]{Proposition}
\aliascntresetthe{proposition}
\newaliascnt{corollary}{theorem}
\newtheorem{corollary}[corollary]{Corollary}
\aliascntresetthe{corollary}

\theoremstyle{definition}
\newaliascnt{definition}{theorem}
\newtheorem{definition}[definition]{Definition}
\aliascntresetthe{definition}
\newaliascnt{remark}{theorem}
\newtheorem{remark}[remark]{Remark}
\aliascntresetthe{remark}
\newaliascnt{example}{theorem}
\newtheorem{example}[example]{Example}
\aliascntresetthe{example}

\newcommand{\R}{\mathbb R}
\newcommand{\conv}{\operatorname{conv}}
\newcommand{\own}{o}
\newcommand{\noindentparagraph}[1]{\par\addvspace{0.8\baselineskip}\noindent#1\enspace\ignorespaces}

\title{GSP Autobidding Equilibria with Few Advertisers or Slots}
\author{%
  Yiyang Huang\thanks{Beijing Institute of Technology, \texttt{huangyiyang@bit.edu.cn}}
  \and Zhengyang Liu\thanks{Beijing Institute of Technology, \texttt{zhengyang@bit.edu.cn}}
  \and Yonglei Yan\thanks{Beijing Institute of Technology, \texttt{yanyonglei@bit.edu.cn}}
  \and Zihe Wang\thanks{Renmin University of China, \texttt{wang.zihe@ruc.edu.cn}}
}
\date{}

\begin{document}
\maketitle

\begin{abstract}
In sponsored-shopping auctions, one advertiser may win several positions, and a product that is never displayed may determine the prices of several winners. Under generalized second-price (GSP) pricing without self-pricing, these interactions couple allocation and payment through a single global ranking. We prove that exact autobidding equilibria under budget and return-on-investment (ROI) constraints can be computed in polynomial time whenever either the number $n$ of advertisers or the number $K$ of display slots is fixed. Our algorithms rely on two structural reductions. With few advertisers, the payment polytopes of tied bid layers combine by a Minkowski sum, giving a fixed-dimensional representation of exponentially many rankings. With few slots, budget equilibria admit threshold recovery and an exact submodular description of attainable payments at the display boundary. For ROI, we prove that every group-consecutive equilibrium has a representative with at most $K+2$ nonunit multipliers, preserving the item and price in every displayed slot. In each case, the reduction leads to polynomially many systems of polynomial equations and inequalities in bounded dimension. The algorithms compute exact multipliers and explicit distributions over valid group-consecutive rankings, with support at most $n+1$; for ROI with few slots, this improves to $\min\{n+1,K+3\}$.
\end{abstract}

\noindent\textbf{Keywords:} Autobidding, generalized second-price auctions, budget constraints, return on investment, equilibrium computation.

\section{Introduction}
\label{sec:introduction}

A product need not be displayed to determine the price of a sponsored position. In a shopping auction, the same advertiser may submit many products and win several positions, while its own products are excluded from setting one another's prices. Under generalized second-price (GSP) pricing without self-pricing, a displayed product pays the highest lower-ranked bid from a different advertiser. That bid may lie far beyond the last display slot. A single undisplayed product can therefore set the prices of several displayed products, as \Cref{ex:tail-pricing} illustrates. This dependence on the entire ranking is the central obstacle to computing an equilibrium.

Sponsored shopping extends the classical position-auction model of sponsored search~\cite{EOS07,Varian07} to advertisers with entire product catalogs~\cite{DLPST26}. It connects the ranking problems of web marketplaces~\cite{ZhangWMMHW18,WuHHL18,GeZLI21} with the economic interaction among automated bidders. In the standard multiplier model, an advertiser delegates bidding to an agent that scales all product values by one multiplier, subject to a budget or a return-on-investment (ROI) constraint~\cite{ABM19,BalseiroDMMZ21,BalseiroBF24}. Computing the resulting equilibrium requires both a multiplier profile and a distribution over bid-consistent rankings.

D\"utting et al.~\cite{DLPST26} established ROI equilibrium existence and welfare guarantees for sponsored-shopping auctions. Hansen, Li, and Nie~\cite{HLN26} subsequently proved existence for budget equilibria (ABE-B) and ROI equilibria (ABE-R), placed exact computation in FIXP, and showed that approximate computation is PPAD-hard in general. These results make the search for tractable structure particularly natural: \emph{can a small number of advertisers or display slots make exact equilibrium computation possible in polynomial time?}

Exact algorithms for second-price pacing equilibria with few buyers~\cite{YWL26} or few goods~\cite{HYWL26} provide a useful starting point. They partition a low-dimensional parameter space into comparison cells and solve linear feasibility problems within each cell. A sponsored-shopping cell, however, may contain exponentially many tied group orders with different payments. Fixing the number of slots creates another difficulty: arbitrarily many advertisers may tie at the display boundary, so neither the number of ranking choices nor the number of payment coordinates becomes bounded. We overcome these obstacles by identifying the outcome geometry that survives each structural restriction.

\noindentparagraph{\textbf{\textup{Principle I: Layerwise Convexification.}}}
With $n$ advertisers, each tied bid layer contains at most $n$ value groups. Although the orders of different layers give exponentially many complete rankings, each layer contributes a payment matrix determined by its own order. The attainable total matrices form a Minkowski sum of local polytopes. In fixed dimension, we can enumerate its vertices and associate a valid ranking with each one. For budgets, we then compress the expected payment vector; for ROI, we include values in the local representation and compress expected surplus. This converts the large family of rankings into a polynomial family of representative outcomes.

\noindentparagraph{\textbf{\textup{Principle II: The Display Boundary.}}}
With $K$ display slots, the displayed--undisplayed boundary determines which parts of a large market must remain in the search. For budgets, the $K$th highest bid recovers every multiplier except those of at most $K-1$ above-threshold advertisers. We characterize the attainable payments of an arbitrarily large threshold tie as a submodular base polytope, then eliminate its payment coordinates by intersecting this polytope with the remaining budget intervals. For ROI, we prove a different structural result: every group-consecutive equilibrium can be transformed into one with at most $K+2$ nonunit multipliers, preserving the item and price in every displayed slot. Two below-threshold price witnesses, or one universal threshold witness, suffice to preserve the relevant tail information. This sparsity reduces the ROI multiplier search to bounded dimension.

The two reductions yield the following theorem. \Cref{fig:reductions} summarizes how the budget and ROI constructions use them.

\begin{theorem}[Exact Autobidding Equilibria in Structurally Small GSP Markets]
\label{thm:main}
For each fixed number $n$ of advertisers or fixed number $K$ of display slots, an exact ABE-B and an exact ABE-R can be computed in polynomial time in the explicit input size. Each algorithm returns exact algebraic multipliers and an explicit probability distribution over valid group-consecutive rankings. The same conclusion holds when the total number $M$ of items is fixed, because $n\leq M$.
\end{theorem}

\Cref{tab:overview} summarizes the four results. Each reduction yields polynomially many candidate systems with boundedly many variables and degree at most two. Standard exact polynomial feasibility then completes the computation.

\begin{table}[t]
\centering
\small
\caption{Polynomial-Time Exact Equilibrium Computation. Support bounds count complete group-consecutive rankings in the equilibrium distribution.}
\label{tab:overview}
\vspace{0.5em}
\begin{tabularx}{\linewidth}{@{}>{\raggedright\arraybackslash}p{0.12\linewidth}>{\raggedright\arraybackslash}p{0.11\linewidth}>{\raggedright\arraybackslash}X>{\centering\arraybackslash}p{0.20\linewidth}>{\centering\arraybackslash}p{0.16\linewidth}@{}}
\toprule
\textbf{Fixed} & \textbf{Model} & \textbf{Structural Reduction} & \textbf{Time} & \textbf{Support} \\
\midrule
Advertisers $n$ & ABE-B & Layerwise payment-matrix convexification & $O(M^{n^3+O(n^2)})$ & $n+1$ \\
\addlinespace[2pt]
Advertisers $n$ & ABE-R & Augmented payment--value convexification & $O(M^{n^3+O(n^2)})$ & $n+1$ \\
\addlinespace[2pt]
Slots $K$ & ABE-B & Threshold recovery and a submodular capacity polytope & $O(M^{O(K^3)})$ & $n+1$ \\
\addlinespace[2pt]
Slots $K$ & ABE-R & At most $K+2$ nonunit multipliers; surplus compression & $O(M^{O(K^2)})$ & $\min\{n+1,\allowbreak K+3\}$ \\
\bottomrule
\end{tabularx}
\end{table}

The support bounds follow from convex-hull compression after the structural reductions have produced a polynomial representation of the attainable outcomes. In the fixed-slot ROI case, multiplier sparsity also reduces the surplus dimension, yielding at most $K+3$ rankings even when the number of advertisers grows.

\subsection{Context and Related Work}
\label{sec:related-work}

\noindentparagraph{\textbf{\textup{Sponsored-Shopping Equilibria.}}}
D\"utting et al.~\cite{DLPST26} introduced the multi-product sponsored-shopping model and proved equilibrium existence and a tight price of anarchy of two for ROI-constrained bidders under GSP and VCG. For the GSP rule studied here, Hansen, Li, and Nie~\cite{HLN26} established existence for both budgets and ROI, polynomial-size equilibrium support, FIXP membership for exact computation, and PPAD-completeness for approximation. We use their group-consecutive existence theorem~\cite[Corollary~4.1]{HLN26} to construct polynomial-time exact algorithms when the advertiser or slot dimension is fixed.

\noindentparagraph{\textbf{\textup{Exact Pacing Algorithms with Few Buyers or Goods.}}}
Multiplicative pacing equilibria have been studied in first- and second-price auction markets~\cite{CKSS22,CKPSSSW22}, with general computational hardness for second-price pacing~\cite{CKK24}. Yan, Wang, and Liu~\cite{YWL26} give an exact algorithm with a fixed number of buyers by partitioning multiplier space into comparison cells and solving linear feasibility problems. Huang et al.~\cite{HYWL26} handle a fixed number of goods by recovering multipliers from highest-bid coordinates, again obtaining linear feasibility within each cell. These works supply two useful starting ideas: comparison cells and recovery from a small set of bid levels. Our fixed-advertiser result additionally represents the coupled outcomes of all tied layers by a payment-matrix Minkowski sum. Our fixed-slot budget result replaces independent per-good allocations by the capacity polytope of a shared threshold layer. The ROI result uses a separate sparsity theorem to obtain a bounded set of multiplier coordinates. These are the structural steps that make the comparison-cell approach work for a global GSP ranking.

\noindentparagraph{\textbf{\textup{Web Advertising and Marketplace Ranking.}}}
The broader literature studies budget allocation and bid optimization in sponsored search~\cite{MehtaSVV07,FeldmanMPS07,BorgsCEIJM07,EvenDarMMMN09}, budget management and pacing dynamics~\cite{BalseiroKMM21,LucierPSZ24}, and auction efficiency with autobidders~\cite{DengMMZ21,DengMMMZZ24,DengGJLM24}. Balseiro et al.~\cite{BalseiroBF24} analyze the interaction of budget and return-on-spend pacing, while Li and Tang~\cite{LT24} study the equilibrium properties and computational difficulties of ROI-constrained second-price markets. Marketplace work addresses relevance-constrained revenue optimization, whole-page allocation, product search, and sponsored-listing re-ranking~\cite{ZhuWYWYC09,ZhangWMMHW18,WuHHL18,GeZLI21,LuZZ26}; audits also document how sponsored results affect search quality and seller visibility~\cite{DashGMCKG24}. Our results address the additional dependence on product ownership and global order in sponsored-shopping equilibrium computation.

\section{Preliminaries}
\label{sec:model}

The budget and ROI models share the ranking and payment rules below.

\subsection{Auction Model}

\noindentparagraph{\textbf{\textup{Auction Model and Bids.}}}
There are $n$ advertisers and $K$ display slots. Write $[r]=\{1,\ldots,r\}$. Advertiser $i\in[n]$ owns $m_i\geq1$ items. Item $(i,j)$ has positive rational value per click $v_{ij}$, with $v_{i1}\geq\cdots\geq v_{i,m_i}$. Let $\mathcal I=\{(i,j):i\in[n],\ j\in[m_i]\}$ be the item set and $M=\sum_{i\in[n]}m_i$ its size. Write $\own(i,j)=i$ for ownership, $u_i=v_{i1}$ for advertiser $i$'s highest value, and $u_{\max}=\max_{i\in[n]}u_i$.

Slot $k\in[K]$ has a rational click-through rate (CTR) $c_k$, where $1\geq c_1\geq\cdots\geq c_K>0$ and $K\leq M$. Set $c_k=0$ for $K<k\leq M$.

Advertiser $i$ uses one multiplier $\alpha_i$ for all its items. The profile $\alpha=(\alpha_1,\ldots,\alpha_n)$ gives bids $b_{ij}(\alpha)=\alpha_i v_{ij}$. In the budget model, it has a positive rational budget $B_i$ and chooses $\alpha_i\in[0,1]$. In the ROI model, $\alpha_i\in[1,A]$, with cap $A>1$ defined below.

\noindentparagraph{\textbf{\textup{Rankings and Payments.}}}
A complete ranking $a=(a(1),\ldots,a(M))$ orders all items exactly once; $a(k)$ is the item ranked $k$th. It is \emph{valid at $\alpha$} if
\[
b_{a(1)}(\alpha)\geq\cdots\geq b_{a(M)}(\alpha).
\]

The first $K$ items occupy the corresponding display slots.

Fix a profile $\alpha$ and a valid ranking $a$. For $a(k)=(i,j)$, the price per click is
\[
\tau_{ij}(a,\alpha)=\max_{\substack{k'>k\\\own(a(k'))\neq i}}b_{a(k')}(\alpha),
\qquad \max\varnothing=0.
\]

This is the highest lower-ranked bid from another advertiser. Undisplayed items remain eligible to set the price.

For $a(k)=(i,j)$, write $v_{a(k)}=v_{ij}$ and $\tau_{a(k)}(a,\alpha)=\tau_{ij}(a,\alpha)$. Advertiser $i$'s value and payment are
\[
V_i(a)=\sum_{\substack{k\in[K]\\\own(a(k))=i}}c_kv_{a(k)},
\quad
P_i(a,\alpha)=\sum_{\substack{k\in[K]\\\own(a(k))=i}}c_k\tau_{a(k)}(a,\alpha),
\]
since $c_k=0$ for $k>K$, undisplayed items yield no value or payment.

\begin{example}[An Undisplayed Product Sets Two Prices]
\label{ex:tail-pricing}
Let $K=2$, $(c_1,c_2)=(1,1/2)$, and all multipliers be one. Advertiser~1 owns products of values $10$ and $8$; advertisers~2 and~3 own one product each, of values $7$ and $6$. Writing each product as item:bid, the ranking is
\[
\underbrace{(1,1):10,\ (1,2):8}_{\text{displayed}}
\quad\big|\quad
\underbrace{(2,1):7,\ (3,1):6}_{\text{undisplayed}}.
\]
Both displayed products pay $7$ per click, giving $P_1=21/2$: advertiser~1's second product is skipped and advertiser~2's undisplayed product sets both prices. If advertiser~2 lowers its bid to $5$, advertiser~3 sets both prices at $6$ and $P_1=9$, with the displayed products unchanged.
\end{example}

\noindentparagraph{\textbf{\textup{Randomized Allocations.}}}
We resolve ties by randomizing over complete rankings. Let $\pi$ be a probability distribution, with $\pi(a)\geq0$ the probability of ranking $a$ and $\sum_a\pi(a)=1$. It is \emph{valid at $\alpha$} if every positive-probability ranking is valid at $\alpha$. Randomization thus preserves bid order and reorders only tied items.

Advertiser $i$'s expected value and payment are
\[
V_i(\pi)=\mathbb E_{a\sim\pi}[V_i(a)],
\qquad
P_i(\pi,\alpha)=\mathbb E_{a\sim\pi}[P_i(a,\alpha)].
\]

Both equilibrium notions use these expected outcomes.

\subsection{Autobidding Equilibria}

We define budget and ROI autobidding equilibria, denoted ABE-B and ABE-R. Both consist of a multiplier profile and a distribution over complete rankings.

The budget model combines feasibility with maximal pacing: an advertiser scales down its bids only when its budget binds in expectation.

\begin{definition}[Autobidding Equilibrium Under Budget Constraints]
For a GSP auction $G=(n,(m_i),K,(v_{ij}),(c_k),(B_i))$, a pair $(\alpha,\pi)$ with $\alpha\in[0,1]^n$ and a probability distribution $\pi$ over complete rankings is an exact ABE-B if:
\begin{enumerate}[label=(\alph*)]
    \item \textbf{Bid-Consistent Allocation.} Every ranking with positive probability under $\pi$ is valid at $\alpha$.
    \item \textbf{Budget Feasibility.} Every advertiser $i\in[n]$ satisfies $P_i(\pi,\alpha)\leq B_i$.
    \item \textbf{Maximal Pacing.} For every advertiser $i\in[n]$, $P_i(\pi,\alpha)<B_i$ implies $\alpha_i=1$.
\end{enumerate}
\end{definition}

This equilibrium notion follows~\cite{HLN26}. Maximal pacing requires unscaled bids when expected payment is below budget. Budget feasibility applies to expected payments under $\pi$.

The ROI model has no fixed budgets. Following~\cite{HLN26}, set the ROI target to one and fix the rational cap
\[
A=1+\frac{\sum_{(i,j)\in\mathcal I}v_{ij}}
{\min_{(i,j)\in\mathcal I}v_{ij}}>1.
\]

Maximal pacing now requires the upper multiplier bound whenever expected value strictly exceeds expected payment.

\begin{definition}[Autobidding Equilibrium Under ROI Constraints]
\label{def:roi}
For a GSP auction $G=(n,(m_i),K,(v_{ij}),(c_k))$, a pair $(\alpha,\pi)$ with $\alpha\in[1,A]^n$ and a probability distribution $\pi$ over complete rankings is an exact ABE-R if:
\begin{enumerate}[label=(\alph*)]
    \item \textbf{Bid-Consistent Allocation.} Every ranking with positive probability under $\pi$ is valid at $\alpha$.
    \item \textbf{ROI Feasibility.} Every advertiser $i\in[n]$ satisfies $V_i(\pi)\geq P_i(\pi,\alpha)$.
    \item \textbf{Maximal Pacing.} For every advertiser $i\in[n]$, $V_i(\pi)>P_i(\pi,\alpha)$ implies $\alpha_i=A$.
\end{enumerate}
\end{definition}

Let $R_i(\pi,\alpha)=V_i(\pi)-P_i(\pi,\alpha)$ be advertiser $i$'s expected surplus. ROI feasibility requires $R_i(\pi,\alpha)\geq0$; maximal pacing requires $R_i(\pi,\alpha)=0$ when $\alpha_i<A$.

\noindentparagraph{\textbf{\textup{Computational Conventions.}}}
All numerical inputs are rational. We use standard exact algorithms for hyperplane arrangements and fixed-dimensional polynomial feasibility~\cite{Edelsbrunner87,BPR06}. Running times are displayed as functions of $M$, suppressing factors polynomial in the size of the numerical input. Multiplicative constants may depend on the fixed parameter; constants hidden in the displayed exponents are universal.

\section{Common Structural Foundations}
\label{sec:structure}

Both structural reductions use positivity of multipliers and existence of group-consecutive equilibria.

\subsection{Multiplier Properties}

We begin with the budget model. Positivity excludes a degenerate boundary of multiplier space and, more importantly, prevents distinct values of one advertiser from collapsing into a common zero-bid layer.

\begin{lemma}[Positive Multipliers in ABE-B]
\label{lem:positive}
Every exact ABE-B $(\alpha,\pi)$ satisfies $\alpha_i>0$ for every advertiser $i\in[n]$.
\end{lemma}

\begin{proof}
In a valid ranking, every lower bid is at most the priced item's bid, so $0\leq\tau_{ij}(a,\alpha)\leq\alpha_i v_{ij}$. If $\alpha_i=0$, advertiser $i$ pays zero in every supported ranking. Since $B_i>0$, maximal pacing would require $\alpha_i=1$, a contradiction.
\end{proof}

Hence, ABE-B multipliers lie in $(0,1]^n$; ABE-R multipliers are positive by definition.

\subsection{Group-Consecutive Rankings}

We restrict the search to group orders using an existence theorem for the two equilibrium models.

Group items if and only if they have the same owner and value. Let $\mathcal G$ be the group set. Group $g\in\mathcal G$ has owner $i(g)$, common value $v_g$, size $q_g$, and common bid $b_g(\alpha)=\alpha_{i(g)}v_g$.

Fix an arbitrary internal item order in each group. A ranking is \emph{group-consecutive} if it orders the groups and expands each group in its fixed internal order. We search only over group orders, without enumerating permutations within a group. The following result ensures that this restriction retains an exact equilibrium.

\begin{corollary}[Hansen--Li--Nie~\cite{HLN26}]
\label{cor:consecutive-existence}
Each model admits an exact equilibrium supported only on rankings in which same-advertiser, same-value items form contiguous blocks.
\end{corollary}

Our definition additionally fixes the internal item order of each group. The next lemma shows that any choice of these internal orders preserves the existence guarantee.

\begin{lemma}[Existence of group-consecutive equilibria]
\label{prop:exist}
Every budget instance admits an exact ABE-B, and every ROI instance with the cap defined above admits an exact ABE-R. In either case, the ranking distribution can be supported on group-consecutive rankings with any prescribed internal item orders.
\end{lemma}

\begin{proof}
Take an equilibrium from \Cref{cor:consecutive-existence} and reorder each consecutive group according to its prescribed internal order. Since items within a group have the same owner and value, the owner, value, and bid at every rank remain unchanged. Thus, all prices and each advertiser's value and payment are preserved.

Combine the probabilities of identical transformed rankings. Under the same multiplier profile, the resulting distribution preserves all expected values and payments, so every equilibrium condition remains satisfied.
\end{proof}

The lemma justifies searching only over group orders. Within each consecutive group, any internal permutation leaves all prices and advertiser outcomes unchanged. We can therefore fix internal orders arbitrarily and ignore them in the search while preserving the equilibrium conditions.

\section{A Common Geometric Framework}
\label{sec:framework}

\Cref{fig:reductions} shows the common structure of the four algorithms. A bounded parameter vector $\theta$ describes the unknown multipliers and, for fixed-slot budgets, a bid threshold. We partition its domain by affine bid comparisons and multiplier bounds. Each relatively open \emph{comparison cell} fixes all strict orders, ties, and endpoint equalities. Within a cell, the structural reduction gives representative rankings or a polyhedral description of their attainable outcomes. Budget conditions depend on expected payments; ROI conditions depend on expected surplus.

Once this outcome set is represented, Carath\'eodory's theorem supplies a common support bound.

\begin{lemma}[Outcome Compression~\cite{Ziegler95}]
\label{lem:outcome-compression}
Fix the parameters and let $\mathcal A$ be a finite set of valid rankings. If a required statistic lies in the convex hull of vectors $s(a)\in\R^d$, $a\in\mathcal A$, then a distribution on at most $d+1$ rankings in $\mathcal A$ realizes the same statistic.
\end{lemma}

The substantive task is therefore to describe the attainable outcomes in polynomial size and bound $d$. With few advertisers, we obtain representative rankings from vertices of a Minkowski sum. With few slots, threshold capacity or multiplier sparsity provides the dimension bound.

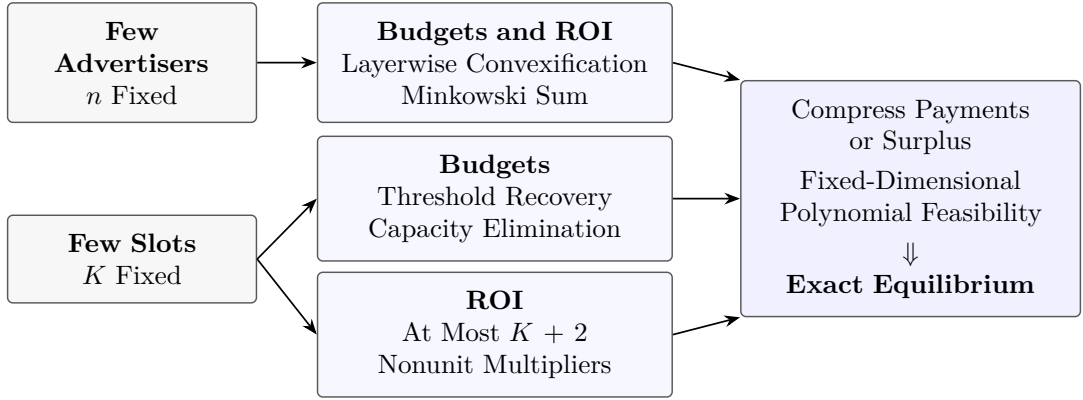
\begin{figure}[htbp]
\centering
\begin{tikzpicture}[
    >=Stealth,
    box/.style={draw=black!65, rounded corners=2pt, align=center,
        inner sep=7pt, font=\small, line width=0.55pt},
    restriction/.style={box, fill=black!3, text width=2.8cm},
    reduction/.style={box, fill=blue!3, text width=4.2cm},
    conclusion/.style={box, fill=blue!6, text width=4.0cm},
    every path/.style={line width=0.65pt}
]
\node[restriction] (advertisers) at (0,2.25)
    {\textbf{Few}\\\textbf{Advertisers}\\$n$ Fixed};
\node[restriction] (slots) at (0,-0.35)
    {\textbf{Few Slots}\\$K$ Fixed};
\node[reduction] (layers) at (4.8,2.25)
    {\textbf{Budgets and ROI}\\Layerwise Convexification\\Minkowski Sum};
\node[reduction] (budgets) at (4.8,0.45)
    {\textbf{Budgets}\\Threshold Recovery\\Capacity Elimination};
\node[reduction] (roi) at (4.8,-1.35)
    {\textbf{ROI}\\At Most $K+2$\\Nonunit Multipliers};
\node[conclusion] (equilibrium) at (10.3,0.45)
    {Compress Payments\\or Surplus\\[3pt]
     Fixed-Dimensional\\Polynomial Feasibility\\[3pt]
     $\Downarrow$\\\textbf{Exact Equilibrium}};
\draw[->] (advertisers) -- (layers);
\draw[->] (slots.east) -- (budgets.west);
\draw[->] (slots.east) -- (roi.west);
\draw[->] (layers.east) -- (equilibrium.north west);
\draw[->] (budgets.east) -- (equilibrium.west);
\draw[->] (roi.east) -- (equilibrium.south west);
\end{tikzpicture}
\caption{Structural Reductions for Exact Equilibrium Computation. Few advertisers admit a common layerwise representation for budgets and ROI. Few slots lead to distinct threshold and sparsity arguments. Each route yields polynomially many equilibrium systems in bounded dimension.}
\label{fig:reductions}
\end{figure}

\noindentparagraph{\textbf{\textup{Solving Candidate Systems.}}}
We apply the exact sampling theorem of Basu, Pollack, and Roy~\cite{BPR06}. A rational system of $m$ polynomial equations and inequalities, including strict inequalities, in $q$ variables and of degree at most two admits a feasibility test and, when feasible, an exact algebraic solution in $(m+1)^{O(q^2)}$ operations. This is polynomial for fixed $q$. In our systems, payments and surplus are expressions in the parameters and ranking probabilities, so the only nonlinear terms are products of a probability and an affine bid. The later analyses establish the required bounds on variables, constraints, and possible supports.

\noindentparagraph{\textbf{\textup{Recovering Weighted Rankings.}}}
An algebraic solution determines the multiplier profile. When the candidate system includes weights for representative group orders, the solution also determines these weights. For fixed-slot budgets with a large threshold tie, we recover a feasible payment vector from the algebraic solution and use \Cref{prop:threshold-decomposition} to express it as a convex combination of group orders; the combination coefficients serve as their probabilities. In either case, expanding the group orders gives complete rankings with precisely the expected payments or surplus specified by the equilibrium system. Combining identical rankings and omitting zero weights preserves the distribution. The constraints of the system then guarantee that the recovered multiplier profile and ranking distribution form an exact equilibrium of the corresponding model.

\section{Layerwise Convexification with Few Advertisers}
\label{sec:advertisers}

With fixed $n$, each tied bid layer contains constantly many groups, but the number of valid complete rankings can still be exponential. For example, with two advertisers and $H$ layers containing one group from each, the independent tie orders give $2^H$ rankings. We represent all of them through local convex hulls and an additive payment formula.

\subsection{Budget Constraints: Compressing Expected Payments}

We begin with budgets, for which preserving the expected payment vector preserves all equilibrium conditions at a fixed multiplier profile.

\begin{theorem}[Exact ABE-B Computation with Few Advertisers]
\label{thm:adv-computation}
For every fixed number of advertisers $n$, there exists a polynomial-time algorithm that computes an exact ABE-B $(\alpha,\pi)$ with time complexity $O(M^{n^3+O(n^2)})$.
\end{theorem}

Algorithm~\ref{alg:advertisers} summarizes the construction in three phases: comparison-cell decomposition, layerwise convexification, and equilibrium construction.

\begin{algorithm}[H]
\small
\caption{Exact ABE-B with a Fixed Number of Advertisers}
\label{alg:advertisers}
\begin{algorithmic}[1]
\REQUIRE An explicit rational budget instance with fixed $n$.
\ENSURE An exact ABE-B with support size at most $n+1$.
\IF{$n=1$}
    \STATE Set $\alpha_1=1$; let $\pi$ select a value-sorted ranking.
    \STATE \textbf{return} $(\alpha,\pi)$.
\ENDIF
\STATE \textbf{Phase I: Comparison-Cell Decomposition.}
\STATE Group items by owner and value; fix internal item orders.
\STATE Enumerate comparison cells, including ties and multiplier endpoints.
\FORALL{cells $F$}
    \STATE \textbf{Phase II: Layerwise Convexification.}
    \STATE Identify bid layers and the paced advertiser set $\mathsf P(F)$.
    \STATE Build local payment matrices and their Minkowski sum $Q_F$.
    \STATE Enumerate the vertices of $Q_F$ and associate a realizing ranking with each vertex.
    \STATE \textbf{Phase III: Equilibrium Construction.}
    \FORALL{nonempty vertex subsets $J$ with $|J|\leq n+1$}
        \STATE Form \Cref{eq:adv-system} for $F$ and $J$.
        \STATE Apply exact semialgebraic sampling.
        \IF{an algebraic solution $(\alpha^*,w^*)$ is found}
            \STATE Expand the associated group orders and assign the sampled weights.
            \STATE Merge identical rankings and remove zero weights.
            \STATE \textbf{return} the resulting exact equilibrium $(\alpha^*,\pi)$.
        \ENDIF
    \ENDFOR
\ENDFOR
\end{algorithmic}
\end{algorithm}

For $n=1$, payments are zero, so multiplier one and a value-sorted ranking give an exact ABE-B. Assume $n\geq2$. \Cref{lem:positive} restricts the search to $(0,1]^n$.

\noindentparagraph{\textbf{\textup{Cell Decomposition.}}}
The order of two cross-owner group bids changes only when $\alpha_{i(g)}/\alpha_{i(g')}=v_{g'}/v_g$. We therefore partition multiplier space by the hyperplanes $v_g\alpha_{i(g)}-v_{g'}\alpha_{i(g')}=0$ for cross-owner groups $g,g'$, together with the endpoint hyperplanes $\alpha_i=1$.

These hyperplanes partition $(0,1]^n$ into nonempty relatively open cells of possibly different dimensions~\cite{Edelsbrunner87}. Let $\mathcal F$ be the cell collection. Include equality faces to represent bid ties and multiplier endpoints. Each cell fixes all bid comparisons and endpoint states; multiplier values may vary.

Each cell $F\in\mathcal F$ also determines the paced advertiser set $\mathsf P(F)=\{i\in[n]:\alpha_i<1\}$. We consider the equilibrium conditions within each cell.

\noindentparagraph{\textbf{\textup{Constructing Bid Layers.}}}
Fix an enumerated cell $F$. Groups $g,g'$ share a bid layer if and only if $b_g(\alpha)=b_{g'}(\alpha)$ for every $\alpha\in F$. Write the layers as $\mathcal G_1,\ldots,\mathcal G_H$ in decreasing bid order.

Value groups depend only on ownership and value, whereas bid layers depend on the cell. Positive multipliers prevent distinct value groups of one advertiser from tying. Thus, each layer contains at most one group per advertiser and admits at most $n!$ group orders.

Fix a group order $\sigma$ in layer $h$. Write $g'\prec_\sigma g$ if $g'$ precedes $g$, and define
\[
s_h=\sum_{h'<h}\sum_{g'\in\mathcal G_{h'}}q_{g'},
\quad
r_g=s_h+\sum_{\substack{g'\in\mathcal G_h\\g'\prec_\sigma g}}q_{g'},
\quad
\gamma_g=\sum_{k=r_g+1}^{r_g+q_g}c_k.
\]

Here, $s_h$ counts items in earlier layers, $r_g$ counts items before $g$, and $\gamma_g$ is $g$'s total click-through rate. We suppress the dependence of $r_g$ and $\gamma_g$ on $\sigma$.

\noindentparagraph{\textbf{\textup{Layerwise Payment Representation.}}}
Let $\Sigma_h$ be layer $h$'s group orders and fix $\sigma\in\Sigma_h$. For a nonlast group $g$, its pricing representative is the next group, which has another owner and the same bid. For the last group, use a fixed rule to choose an other-owner group in the highest lower layer containing one. If none exists, its price is zero. Eligible groups in that layer have equal bids, so their order does not affect the price.

Let $E_{h,\sigma}$ be the $n\times n$ matrix whose $(i,j)$ entry is the sum of $\gamma_gv_{g'}$ over groups $g$ owned by $i$ whose pricing representative $g'$ is owned by $j$. Each contribution gives payment $\gamma_gv_{g'}\alpha_j$, so $E_{h,\sigma}\alpha$ is the layer's payment vector. Its diagonal is zero because pricing excludes the displayed product's owner.

A layer's payments depend only on its group order and $\alpha$. A complete ranking $a$ with layer orders $\sigma_1,\ldots,\sigma_H$ therefore satisfies
\[
\bigl(P_i(a,\alpha)\bigr)_{i\in[n]}
=
\left(\sum_{h=1}^{H}E_{h,\sigma_h}\right)\alpha
\qquad(\alpha\in F).
\]

\noindentparagraph{\textbf{\textup{Combining Layerwise Payments.}}}
A ranking chooses one order per layer. Enumerating all combinations can be exponential in $H$, even for fixed $n$. Instead, represent each layer's randomization by
\[
Q_h=\conv\{E_{h,\sigma}:\sigma\in\Sigma_h\},
\]
and combine layers by the Minkowski sum
\[
Q_F=\sum_{h=1}^{H}Q_h
=
\left\{\sum_{h=1}^{H}D_h:D_h\in Q_h\right\}.
\]

With this construction, only $Hn!$ local matrices are needed. By the additive payment formula above, the expected payment matrix of a distribution over complete rankings is determined by its marginal randomization in each layer. The following lemma shows that $Q_F$ represents every attainable expected payment vector.

\begin{lemma}[Expected Payment-Matrix Representation]
\label{lem:minkowski}
The polytope $Q_F$ is the convex hull of total payment matrices of all group-consecutive rankings valid on $F$:
\[
Q_F=
\conv\left\{
\sum_{h=1}^{H}E_{h,\sigma_h}:
\sigma_h\in\Sigma_h
\right\}.
\]

For fixed $\alpha\in F$, the attainable expected payment vectors are exactly $\{D\alpha:D\in Q_F\}$.
\end{lemma}

\begin{proof}
For a distribution over valid group-consecutive rankings, let $\lambda_{h,\sigma}$ be order $\sigma$'s marginal probability in layer $h$. The expected total matrix is
\[
\sum_{h=1}^{H}
\left(
\sum_{\sigma\in\Sigma_h}
\lambda_{h,\sigma}E_{h,\sigma}
\right).
\]

Each term in parentheses lies in $Q_h$, so the matrix lies in $Q_F$.

Conversely, write $D\in Q_F$ as $\sum_hD_h$, with each $D_h$ a convex combination of local matrices. Independently choose each layer's order with these coefficients as probabilities. Concatenating orders and expanding groups gives valid complete rankings with expected matrix $D$. This proves equality. Multiplying by fixed $\alpha$ gives the payment-vector claim.
\end{proof}

\noindentparagraph{\textbf{\textup{Vertex Enumeration and Representative Rankings.}}}
We enumerate the vertices of $Q_F$ and associate a realizing ranking with each one. Let $\mathcal E_h$ be the distinct local matrices in layer $h$, with one corresponding group order for each matrix.

Identify zero-diagonal matrices with vectors in $\R^d$, where $d=n(n-1)$. An objective $\xi\in\R^d$ assigns score $\langle\xi,E\rangle=\sum_{i\neq j}\xi_{ij}E_{ij}$.

For each layer and distinct $E,E'\in\mathcal E_h$, use the hyperplane $\langle \xi,E-E'\rangle=0$.

Enumerate full-dimensional arrangement regions and choose a rational objective in each. Local score comparisons are strict and fixed on a region, so each layer has a unique maximizer $E_h^*(\xi)$. Set $D^*(\xi)=\sum_{h=1}^{H}E_h^*(\xi)$.

Concatenate associated orders to realize this matrix and remove duplicate total matrices.

Maximizing a linear objective over $Q_F=\sum_hQ_h$ amounts to maximizing it over each summand. Unique local maximizers sum to the unique maximizer over $Q_F$, hence a vertex. Their associated orders give a realizing ranking. Conversely, every vertex uniquely maximizes all objectives in some nonempty open set. This set contains an objective outside all comparison hyperplanes, so its enumerated region produces that vertex.

The objective arrangement has dimension $d=n(n-1)$ and at most $H\binom{n!}{2}$ hyperplanes. Standard fixed-dimensional arrangement enumeration therefore gives all vertices and realizing rankings in polynomial time for fixed $n$.

The distinct vertices $D^{(1)},\ldots,D^{(N_F)}$ and associated rankings $a^{(1)},\ldots,a^{(N_F)}$ form the candidate list.

\noindentparagraph{\textbf{\textup{Support Compression.}}}
At fixed $\alpha\in F$, candidate $\ell$ has payment vector $D^{(\ell)}\alpha\in\R^n$. We can now apply \Cref{lem:outcome-compression} to the representative rankings. This reduces the distribution to at most $n+1$ representative rankings while preserving every advertiser's expected payment and hence all budget equilibrium conditions at the same multipliers.

\begin{lemma}[Small Candidate Support]
\label{lem:vertex-support}
At fixed $\alpha\in F$, every distribution over valid group-consecutive rankings can be replaced by at most $n+1$ representative rankings with the same expected advertiser payments. This preserves an ABE-B at the same profile.
\end{lemma}

\begin{proof}
By \Cref{lem:minkowski} and vertex enumeration, the expected payment vector lies in $\conv\{D^{(\ell)}\alpha:\ell\in[N_F]\}$. Apply \Cref{lem:outcome-compression} with $d=n$. Preserving payments at the same multipliers preserves budget feasibility and maximal pacing.
\end{proof}

Thus, we only enumerate nonempty $J\subseteq[N_F]$ with $|J|\leq n+1$, since every attainable expected payment vector can be preserved on such a support. For $\ell\in J$, assign $a^{(\ell)}$ probability $w_\ell$, with $w_\ell\geq0$ and $\sum_{\ell\in J}w_\ell=1$. We solve these weights and multipliers jointly.

For fixed $n$, there are $O((N_F+1)^{n+1})$ possible supports, so polynomially bounded $N_F$ makes this enumeration polynomial. Together with \Cref{prop:exist}, the support lemma ensures that one of them realizes an exact ABE-B.

\noindentparagraph{\textbf{\textup{Equilibrium Constraints.}}}
Fix a cell $F$ and support $J$. Advertiser $i$'s expected payment is
\[
p_i(\alpha,w)
=
\sum_{\ell\in J}w_\ell
\sum_{j\in[n]}(D^{(\ell)})_{ij}\alpha_j.
\]

Impose the candidate system
\begin{equation}
\label{eq:adv-system}
\left\{
\begin{aligned}
&\alpha\in F,\qquad w_\ell\geq0\quad(\ell\in J),\\
&\sum_{\ell\in J}w_\ell=1,\\
&p_i(\alpha,w)\leq B_i\quad(i\in[n]),\\
&p_i(\alpha,w)=B_i\quad(i\in\mathsf P(F)).
\end{aligned}
\right.
\end{equation}

Every solution defines an exact ABE-B. The remaining algorithmic point is that the number of unknowns is independent of the catalog size.

\begin{lemma}[Polynomial-Time Candidate Feasibility]
\label{lem:adv-system-feasibility}
For fixed $n$, we can decide feasibility of \Cref{eq:adv-system} and compute an exact solution whenever one exists, in polynomial time.
\end{lemma}

\begin{proof}[Proof of \Cref{lem:adv-system-feasibility}]
The $n$ multipliers and $|J|$ weights give at most $2n+1$ variables. There are $O(M^2+n)$ conditions, all affine except payment terms proportional to $w_\ell\alpha_j$, so degree is at most two. Matrix entries are sums of rational value--CTR products. The exact feasibility bound of \Cref{sec:framework} applies.
\end{proof}

\noindentparagraph{\textbf{\textup{Correctness and Running Time.}}}
Recall that $d=n(n-1)$; let $N_{\mathrm{cell}}$ count comparison cells and $N_F$ vertices in cell $F$.

\begin{proof}[Proof of \Cref{thm:adv-computation}]
The representative rankings are valid throughout their cell, and \Cref{eq:adv-system} uses their actual expected payments. Every returned algebraic solution thus satisfies bid consistency, budget feasibility, and maximal pacing.

For completeness, take a group-consecutive ABE-B from \Cref{prop:exist}. Its profile lies in an enumerated cell. \Cref{lem:vertex-support} replaces its distribution by at most $n+1$ representative rankings with the same expected payments. This support is enumerated, so a candidate system is feasible and the algorithm returns an exact ABE-B.

We next bound the running time. The case $n=1$ requires only sorting. Assume $n\geq2$.

\emph{Comparison Cells.} The $O(M^2)$ hyperplanes in dimension $n$ define $N_{\mathrm{cell}}=O(M^{2n})$ cells, including equality faces. Standard arrangement algorithms enumerate them in $M^{O(n)}$ time~\cite{Edelsbrunner87}.

\emph{Local Matrices and Vertices.} There are at most $M$ layers and $n!$ orders per layer. Their matrices take $O(M^2)$ time to construct for fixed $n$. The objective arrangement has $O(M)$ hyperplanes in dimension $d=n(n-1)$, so $N_F=O(M^d)$. The vertices and their realizing rankings can be computed in $M^{O(n^2)}$ time per cell.

\emph{Candidate Supports.} Their number per cell is
\[
\sum_{r=1}^{\min\{n+1,N_F\}}\binom{N_F}{r}
=
O\bigl((N_F+1)^{n+1}\bigr)
=
O\bigl(M^{d(n+1)}\bigr).
\]

\emph{Candidate Systems.} Each has at most $2n+1$ variables, $O(M^2+n)$ conditions, and degree at most two. By the exact feasibility bound in \Cref{sec:framework}, each system can be constructed and solved in $M^{O(n^2)}$ time.

Combining these bounds gives
\[
\mathcal T_{\mathrm{total}}
=
O\bigl(M^{2n+d(n+1)+O(n^2)}\bigr)
=
O\bigl(M^{n^3+O(n^2)}\bigr),
\]
since $2n+d(n+1)=n^3+n$. Constructing at most $n+1$ complete rankings takes $O(M)$ additional time for fixed $n$, within this bound.
\end{proof}

\subsection{ROI Constraints: Compressing Expected Surplus}
\label{sec:roi-advertisers}

The comparison-cell decomposition, bid-layer construction, and vertex-enumeration method follow the same framework as ABE-B. For ROI, we augment each local payment matrix with its value vector and compress the resulting surplus vectors.

\begin{theorem}[Exact ABE-R Computation with Few Advertisers]
\label{thm:roi-computation}
For every fixed number of advertisers $n$, there exists a polynomial-time algorithm that computes an exact ABE-R $(\alpha,\pi)$ with time complexity $O(M^{n^3+O(n^2)})$.
\end{theorem}

The output distribution uses at most $n+1$ group-consecutive rankings.

Algorithm~\ref{alg:roi-advertisers} follows the same three-phase structure as Algorithm~\ref{alg:advertisers}, using pairs of payment matrices and value vectors and surplus constraints in place of payment matrices and budget constraints.

For $n=1$, set $\alpha_1=A$ and let $\pi$ select a value-sorted ranking with probability one. Payments are zero, and maximal pacing holds at the cap. Assume $n\geq2$.

\noindentparagraph{\textbf{\textup{Cell Decomposition and Augmented Representation.}}}
As in ABE-B, compare all cross-owner group bids, now on $[1,A]^n$, and include both endpoint forms $\alpha_i-1$ and $\alpha_i-A$. Include equality faces. A cell $F$ fixes all bid comparisons and the set $\mathsf C(F)=\{i:\alpha_i<A\text{ throughout }F\}$. Advertisers in this set must have zero expected surplus; those at the cap need only have nonnegative expected surplus. Build bid layers and local payment matrices as in ABE-B.

For each $\sigma\in\Sigma_h$, also record value vector $W_{h,\sigma}$ with coordinates
\[
W_{h,\sigma,i}
=
\sum_{\substack{g\in\mathcal G_h\\i(g)=i}}
\gamma_gv_g.
\]

Values and payments add over layers, so define
\[
Q_F^R
=
\sum_{h=1}^{H}
\conv\left\{
(E_{h,\sigma},W_{h,\sigma}):
\sigma\in\Sigma_h
\right\}.
\]

The argument of \Cref{lem:minkowski} identifies this polytope with the convex hull of pairs of payment matrices and value vectors for valid group-consecutive rankings on $F$. Its dimension is at most $n(n-1)+n=n^2$, so the same objective-arrangement construction enumerates its $N_F$ vertices. For each vertex $(D^{(\ell)},W^{(\ell)})$, let $a^{(\ell)}$ be the complete ranking defined by the corresponding layer orders. Then
\[
V_i(a^{(\ell)})=W_i^{(\ell)},
\qquad
P_i(a^{(\ell)},\alpha)=(D^{(\ell)}\alpha)_i
\quad(\alpha\in F).
\]

\begin{algorithm}[H]
\small
\caption{Exact ABE-R with a Fixed Number of Advertisers}
\label{alg:roi-advertisers}
\begin{algorithmic}[1]
\REQUIRE An explicit rational ROI instance with fixed $n$.
\ENSURE An exact ABE-R with support size at most $n+1$.
\STATE Compute the cap $A$.
\IF{$n=1$}
    \STATE Set $\alpha_1=A$; let $\pi$ select a value-sorted ranking.
    \STATE \textbf{return} $(\alpha,\pi)$.
\ENDIF
\STATE \textbf{Phase I: Comparison-Cell Decomposition.}
\STATE Group items by owner and value; fix internal item orders.
\STATE Enumerate comparison cells in $[1,A]^n$, including ties and multiplier endpoints.
\FORALL{cells $F$}
    \STATE \textbf{Phase II: Layerwise Convexification.}
    \STATE Identify bid layers and the below-cap advertiser set $\mathsf C(F)$.
    \STATE Build local pairs of payment matrices and value vectors and their Minkowski sum $Q_F^R$.
    \STATE Enumerate the vertices of $Q_F^R$ and associate a realizing ranking with each vertex.
    \STATE \textbf{Phase III: Equilibrium Construction.}
    \FORALL{nonempty vertex subsets $J$ with $|J|\leq n+1$}
        \STATE Form \Cref{eq:roi-advertiser-system} for $F$ and $J$.
        \STATE Apply exact semialgebraic sampling.
        \IF{an algebraic solution $(\alpha^*,w^*)$ is found}
            \STATE Expand the associated group orders and assign the sampled weights.
            \STATE Merge identical rankings and remove zero weights.
            \STATE \textbf{return} the resulting exact equilibrium $(\alpha^*,\pi)$.
        \ENDIF
    \ENDFOR
\ENDFOR
\end{algorithmic}
\end{algorithm}

\noindentparagraph{\textbf{\textup{Surplus Compression and Equilibrium Constraints.}}}
Fix $\alpha\in F$. The linear map $(D,W)\mapsto W-D\alpha$ sends $Q_F^R$ to the attainable expected surplus vectors. Hence, these vectors form the convex hull of $W^{(\ell)}-D^{(\ell)}\alpha$, $\ell\in[N_F]$. \Cref{lem:outcome-compression} with $d=n$ gives at most $n+1$ representative rankings. Preserving surplus preserves all ROI conditions at the same multipliers, even when values and payments change separately.

Enumerate every nonempty $J\subseteq[N_F]$ with $|J|\leq n+1$. Assign probability $w_\ell$ to ranking $a^{(\ell)}$ and define expected surplus
\begin{equation}
\label{eq:roi-adv-surplus}
r_i(\alpha,w)
=
\sum_{\ell\in J}
w_\ell\left(
W_i^{(\ell)}-(D^{(\ell)}\alpha)_i
\right).
\end{equation}

ROI feasibility and maximal pacing are precisely the constraints
\begin{equation}
\label{eq:roi-advertiser-system}
\left\{
\begin{aligned}
&\alpha\in F,\\
&w_\ell\geq0\quad(\ell\in J),\qquad
\sum_{\ell\in J}w_\ell=1,\\
&r_i(\alpha,w)\geq0\quad(i\in[n]),\\
&r_i(\alpha,w)=0\quad(i\in\mathsf C(F)).
\end{aligned}
\right.
\end{equation}

Every solution defines an exact ABE-R. The feasibility argument of \Cref{lem:adv-system-feasibility} applies here as well: the $n$ multipliers and at most $n+1$ weights give at most $2n+1$ variables, there are $O(M^2+n)$ conditions, and the only nonlinear terms are products $w_\ell\alpha_j$. Thus, for fixed $n$, the exact feasibility bound in \Cref{sec:framework} decides feasibility of \Cref{eq:roi-advertiser-system} and returns an exact algebraic solution whenever one exists, in polynomial time.

\noindentparagraph{\textbf{\textup{Recovering the Equilibrium.}}}
Given an algebraic solution $(\alpha^*,w^*)$, concatenate each representative ranking's layer orders and apply the common recovery procedure in \Cref{sec:framework}. The defining value-payment identities give $R_i(\pi,\alpha^*)=r_i(\alpha^*,w^*)$, so \Cref{eq:roi-advertiser-system} enforces exactly the ROI conditions. At most $n+1$ rankings are expanded, taking $O((n+1)M)$ operations.

\noindentparagraph{\textbf{\textup{Correctness and Running Time.}}}
\begin{proof}[Proof of \Cref{thm:roi-computation}]
Each representative ranking is valid throughout its cell, and \Cref{eq:roi-adv-surplus} gives its contribution to expected surplus. Every algebraic solution thus satisfies all ABE-R conditions.

For completeness, take a group-consecutive ABE-R from \Cref{prop:exist}. Its profile lies in an enumerated cell, and the pair formed by its expected payment matrix and value vector lies in $Q_F^R$. At this profile, its surplus vector lies in
\[
\conv\left\{
W^{(\ell)}-D^{(\ell)}\alpha:\ell\in[N_F]
\right\}.
\]

\Cref{lem:outcome-compression} uses at most $n+1$ associated rankings. This support is enumerated, so a candidate system is feasible.

For fixed $n$, there are $O(M^{2n})$ cells. The objective arrangement has $O(M)$ hyperplanes in dimension $n^2$, giving $N_F=O(M^{n^2})$ vertices and $O(M^{n^2(n+1)})$ supports per cell. Each system has at most $2n+1$ variables, $O(M^2+n)$ conditions, and degree at most two; its only nonlinear terms are $w_\ell\alpha_j$. The exact-sampling bound used for ABE-B gives $O(M^{O(n^2)})$ operations per system. Thus, the total operation count is
\[
O\bigl(M^{2n+n^2(n+1)+O(n^2)}\bigr)
=
O\bigl(M^{n^3+O(n^2)}\bigr).
\]

Vertex construction and output fit the same bounds.
\end{proof}

\section{Boundary Structure with Few Slots}
\label{sec:positions}

Unlike the fixed-advertiser setting in \Cref{sec:advertisers}, fixing the number $K$ of display slots does not bound the number $n$ of advertisers. Directly reusing the full $n$-dimensional multiplier-space decomposition would therefore give an exponent depending on $n$, rather than a polynomial-time bound for unrestricted $n$. A tied bid layer may still contain arbitrarily many advertisers, and undisplayed items may still determine displayed prices, so restricting attention to the displayed prefix is insufficient. We therefore work from the display boundary to identify the multiplier coordinates and expected outcomes needed by each constraint model, with the aim of bounding the number of unknowns in terms of $K$.

\subsection{Budget Constraints: Threshold Recovery and Capacity}

For budgets, maximal pacing pins down every advertiser that lies completely below the display threshold. Combining this recovery identity with a polyhedral description of a potentially large threshold tie yields the fixed-slot result.

\begin{theorem}[Exact ABE-B Computation with Few Slots]
\label{thm:pos-computation}
For every fixed number of display slots $K$, there exists a polynomial-time algorithm that computes an exact ABE-B $(\alpha,\pi)$ with time complexity $O(M^{O(K^3)})$.
\end{theorem}

The key parameter is the $K$th highest bid, which determines all but at most $K-1$ multipliers. Payments can then be described by a bounded family of rankings or, for a large threshold tie, by a capacity polytope. Algorithm~\ref{alg:positions} summarizes the construction in three phases: threshold recovery, attainable payments, and equilibrium construction.

\begin{algorithm}[H]
\small
\caption{Exact ABE-B with a Fixed Number of Display Slots}
\label{alg:positions}
\begin{algorithmic}[1]
\REQUIRE An explicit rational budget instance with fixed $K$.
\ENSURE An exact ABE-B with support size at most $n+1$.
\STATE \textbf{Phase I: Threshold Recovery.}
\STATE Compute value groups and highest values $u_i$.
\FORALL{$S\subseteq[n]$ with $|S|\leq K-1$ and threshold pieces $I$}
    \STATE Parameterize multipliers by $(t,\alpha_S)$ and enumerate comparison cells $F$.
    \FORALL{admissible cells $F$}
        \STATE \textbf{Phase II: Attainable Payments.}
        \STATE Enumerate valid orders strictly above $t$ and compute the number $e$ of threshold groups, the number $r$ of remaining slots, and the threshold-owner set $T$.
        \IF{$e\leq r$}
            \STATE Enumerate the bounded family of threshold-layer orders and form \Cref{eq:pos-few-system}.
        \ELSE
            \STATE Form the submodular capacity conditions in \Cref{eq:pos-many-system}.
        \ENDIF
        \STATE \textbf{Phase III: Equilibrium Construction.}
        \FORALL{candidate supports produced in Phase II}
            \STATE Apply exact semialgebraic sampling.
            \IF{an algebraic solution is found}
                \STATE Recover the full multiplier profile, decompose threshold payments, and compress the ranking distribution.
                \STATE \textbf{return} the resulting exact equilibrium.
            \ENDIF
        \ENDFOR
    \ENDFOR
\ENDFOR
\end{algorithmic}
\end{algorithm}

\noindentparagraph{\textbf{\textup{Threshold Recovery.}}}
Recall that $u_i$ is advertiser $i$'s highest value. Let $b_{(1)}(\alpha)\geq\cdots\geq b_{(M)}(\alpha)$ be the ordered item bids, counting every item, and define $t=b_{(K)}(\alpha)$. Thus, $t$ is the last displayed bid, independent of tie order. Define $S=\{i\in[n]:\alpha_i u_i>t\}$ and $s=|S|$. These are precisely the advertisers with a product bidding strictly above $t$. By \Cref{lem:positive} in \Cref{sec:structure}, all equilibrium multipliers are positive. Since item values are positive, $t>0$.

The threshold becomes algorithmically useful because it determines every multiplier except those of a bounded set of above-threshold advertisers.

\begin{lemma}[Threshold Recovery]
\label{lem:pos-recovery}
Every exact ABE-B satisfies
\begin{equation}
\label{eq:pos-recovery}
|S|\leq K-1,
\qquad
\alpha_i=\min\{1,t/u_i\}
\quad(i\notin S).
\end{equation}
\end{lemma}

\begin{proof}
At most $K-1$ items bid strictly above $t$. Each advertiser in $S$ owns one, so $|S|\leq K-1$.

For $i\notin S$, we have $\alpha_i u_i\leq t$ and $\alpha_i\leq1$, so $\alpha_i\leq\min\{1,t/u_i\}$. A strict inequality would give $\alpha_i<1$ and $\alpha_i u_i<t$. All its items would bid below $t$ and be undisplayed in every valid ranking, giving $P_i(\pi,\alpha)=0$. Since $B_i>0$, maximal pacing would require $\alpha_i=1$, a contradiction. Thus, the recovery formula holds.
\end{proof}

\noindentparagraph{\textbf{\textup{From Recovery to a Search Space.}}}
Enumerate $S\subseteq[n]$ with $s=|S|\leq K-1$ as candidate sets of advertisers whose highest bids exceed $t$. For each $S$, keep $\theta=(t,\alpha_S)$ unknown and recover other multipliers by \Cref{eq:pos-recovery}. Require $0<\alpha_i\leq1$ and $\alpha_i u_i>t$ for $i\in S$. Recovered multipliers satisfy $\alpha_i u_i\leq t$ outside $S$.

Partition $(0,u_{\max}]$ at distinct $u_i$, keeping open intervals and singleton endpoints. For each piece $I$ and pair $(S,I)$, impose $t\in I$. Each recovered multiplier then has a fixed affine form, $t/u_i$ or $1$. Both forms agree at $t=u_i$. Thus, every group bid $b_g(\theta)=v_g\alpha_{i(g)}(\theta)$ is affine in the at most $K$ parameters $\theta$.

We partition this parameter space as in \Cref{sec:advertisers}, using cross-owner bid comparisons, comparisons with the threshold, and multiplier endpoints:
\[
b_g(\theta)-b_{g'}(\theta)=0,
\quad
b_g(\theta)=t,
\quad
\alpha_i=1.
\]
Here $g,g'$ range over cross-owner groups, $g$ ranges over all groups in the threshold comparisons, and $i\in S$. Constant expressions have fixed signs, whereas nonconstant expressions define the comparison hyperplanes.

Let $\mathcal F(S,I)$ be the nonempty sign cells, including equality faces. Fix $F\in\mathcal F(S,I)$. For $\diamond\in\{>,=,<\}$ and any $\theta\in F$, define
\[
\mathcal G_{\diamond}(F)
=
\{g\in\mathcal G:b_g(\theta)\diamond t\}.
\]

Membership is constant on $F$ because comparisons with $t$ have fixed signs. These sets partition $\mathcal G$. Every valid ranking places $\mathcal G_{>}(F)$ before $\mathcal G_{=}(F)$ and $\mathcal G_{=}(F)$ before $\mathcal G_{<}(F)$.

A fixed $(S,I,F)$ determines all group-bid comparisons and hence the valid group orders, with ties free to reorder. Endpoint states follow from $F$ inside $S$ and $I$ outside $S$, fixing the full paced set $\mathsf P(F)$. Parameter values remain unknown and are solved jointly with ranking probabilities.

\noindentparagraph{\textbf{\textup{Admissible Threshold Cells.}}}
For fixed $F\in\mathcal F(S,I)$, define
\[
h=\sum_{g\in\mathcal G_{>}(F)}q_g,
\qquad
T=\{i(g):g\in\mathcal G_{=}(F)\},
\]
and set $e=|T|$ and $r=K-h$. Here, $h$ counts items in $\mathcal G_{>}(F)$ and $T$ lists owners in $\mathcal G_{=}(F)$. Positive multipliers allow at most one group per owner in $\mathcal G_{=}(F)$, so $e=|\mathcal G_{=}(F)|$. For $i\in T$, let $q_i$ be its unique group's size. A cell is admissible if
\[
h\leq K-1,
\qquad
\sum_{g\in\mathcal G_{=}(F)}q_g\geq r.
\]

These conditions are equivalent to $t$ being the $K$th highest bid throughout the cell: fewer than $K$ items bid strictly above $t$, and at least $K$ bid at least $t$. We consider only admissible cells.

\noindentparagraph{\textbf{\textup{Candidate Rankings and Payment Representation.}}}
Fix an admissible $(S,I,F)$ with $F\in\mathcal F(S,I)$. All bid comparisons are fixed, so only tied groups may reorder. Enumerate these orders in $\mathcal G_{>}(F)$, preserving strict comparisons. Fix one valid order of $\mathcal G_{<}(F)$ for all candidates. Advertisers outside $S\cup T$ have multiplier one and payment zero.

\begin{remark}[Ordering Groups Below $t$]
Fix $\theta\in F$ and the orders of $\mathcal G_{>}(F)$ and $\mathcal G_{=}(F)$. As with fixed internal item orders, choosing any valid order of $\mathcal G_{<}(F)$ preserves displayed outcomes. The only concern is that these items may still set displayed prices. The highest lower-tail bid from another owner may price displayed items in the last group of $\mathcal G_{=}(F)$. Different valid tail orders may change which group supplies this bid, but not the bid itself, and if no eligible group exists, its contribution remains zero.
\end{remark}

Compare $e=|\mathcal G_{=}(F)|$ with $r$ to decide whether to enumerate group orders directly. If $e\leq r$, at most $h+e\leq h+r=K$ groups lie in $\mathcal G_{>}(F)\cup\mathcal G_{=}(F)$, so we can list valid orders and compute their payments. If $e>r$, listing all orders of $\mathcal G_{=}(F)$ may be too costly. At most $r$ of these $e$ groups can receive a display slot, so at least one whole group remains undisplayed, with all its items bidding $t$. This group has a different owner from every displayed group in $\mathcal G_{=}(F)$. Payments contributed by this set are therefore $t$ times each advertiser's assigned total click-through rate. This lets us describe expected payments without listing all orders of $\mathcal G_{=}(F)$. We first treat $e\leq r$.

Suppose $e\leq r$. Concatenate each valid order of $\mathcal G_{>}(F)$ with each order of $\mathcal G_{=}(F)$, then append the fixed $\mathcal G_{<}(F)$ order. This gives at most $h!e!\leq K!$ complete candidates. Only $S\cup T$ can pay, so payment vectors have at most $a_F:=|S\cup T|\leq h+e\leq K$ coordinates. Write candidate $\ell$'s affine payments as $p_i^{(\ell)}(\theta)$.

At fixed $\theta$, \Cref{lem:outcome-compression} with $d=a_F$ uses at most $a_F+1$ candidates. Enumerate nonempty candidate index sets $J$ with $|J|\leq a_F+1$, introduce probabilities $w_\ell$, and define
\[
\overline p_i(\theta,w)
=
\sum_{\ell\in J}w_\ell p_i^{(\ell)}(\theta).
\]

For each support $J$, impose
\begin{equation}
\label{eq:pos-few-system}
\left\{
\begin{aligned}
&\theta\in F,\quad w_\ell\geq0\quad(\ell\in J),\\
&\sum_{\ell\in J}w_\ell=1,\\
&\overline p_i(\theta,w)\leq B_i
&& (i\in S\cup T),\\
&\overline p_i(\theta,w)=B_i
&& (i\in(S\cup T)\cap\mathsf P(F)).
\end{aligned}
\right.
\end{equation}

There are at most $(s+1)+(a_F+1)\leq2K+1$ variables, and the degree is at most two. Exact feasibility and the ranking construction in \Cref{sec:framework} give an ABE-B. The representative rankings include all below-threshold price setters.

\noindentparagraph{\textbf{\textup{Payments Above $t$ and Support Compression.}}}
Fix an admissible $(S,I,F)$ with $e>r$. A direct threshold-order representation may contain $e!$ rankings, and its payment dimension $|S\cup T|$ is unbounded in $K$. We use a short distribution over above-threshold orders and eliminate threshold payments through capacity constraints.

Recall $\theta=(t,\alpha_S)$ and $s=|S|$. Distinct bid layers have fixed order on $F$, so only equal-bid groups may swap. Let $\Sigma_{>}(F)$ contain the valid group orders of $\mathcal G_{>}(F)$. For $\rho\in\Sigma_{>}(F)$, let $p_i^{+,\rho}(\theta)$ be advertiser $i$'s payment from items in $\mathcal G_{>}(F)$ under $\rho$.

These functions are affine in $\theta$. Each displayed item's price is an eligible later bid in $\mathcal G_{>}(F)$ or $t$, since $\mathcal G_{=}(F)$ has at least two owners. Thus, orders within $\mathcal G_{=}(F)$ and $\mathcal G_{<}(F)$ do not affect these payments. They vanish outside $S$, as only $S$ owns items in $\mathcal G_{>}(F)$.

Only $s$ upper-payment coordinates can be nonzero, so \Cref{lem:outcome-compression} gives at most $s+1$ orders. Enumerate nonempty supports $J\subseteq\Sigma_{>}(F)$ with $|J|\leq s+1$, introduce probabilities $\beta_\rho$, and define
\[
p_i^+(\theta,\beta)
=
\sum_{\rho\in J}
\beta_\rho p_i^{+,\rho}(\theta).
\]

Thus, $p_i^+$ is advertiser $i$'s expected payment from $\mathcal G_{>}(F)$ alone, and is zero for $i\notin S$. It depends on $\theta$ and $\beta$ and adds no variables.

\noindentparagraph{\textbf{\textup{Attainable Payments at $t$.}}}
Since $e>r$, at most $r$ groups in $\mathcal G_{=}(F)$ are displayed, leaving at least one wholly undisplayed. Each displayed item in this set thus has an other-owner item bidding $t$ below it and pays $t$ per click. Each advertiser's payment from $\mathcal G_{=}(F)$ is therefore $t$ times its assigned total click-through rate.

Let $\Sigma_{=}(F)$ contain all group orders of $\mathcal G_{=}(F)$. For $\sigma\in\Sigma_{=}(F)$, place these groups after $\mathcal G_{>}(F)$ to fill slots $h+1,\ldots,K$. Let $x_i^\sigma$ be advertiser $i$'s total click-through rate under this order. With price per click $t$, define the payment vector by
\[
y_i^\sigma(t)=t x_i^\sigma,
\qquad
y^\sigma(t)=\bigl(y_i^\sigma(t)\bigr)_{i\in T}.
\]

Assigning probability $\lambda_\sigma$ to order $\sigma$ gives expected payment
\[
y=\sum_{\sigma\in\Sigma_{=}(F)}\lambda_\sigma y^\sigma(t),
\qquad
\lambda_\sigma\geq0,
\qquad
\sum_{\sigma\in\Sigma_{=}(F)}\lambda_\sigma=1.
\]

The attainable expected payment vectors form
\[
\mathcal Y_F(t)
=
\conv
\{y^\sigma(t):\sigma\in\Sigma_{=}(F)\}.
\]

For $U\subseteq T$, define
\[
q(U)=\sum_{i\in U}q_i,
\qquad
\Phi(U)=\sum_{k=h+1}^{h+\min\{r,q(U)\}}c_k.
\]

Here, $q(U)$ counts items in $\mathcal G_{=}(F)$ owned by $U$. They occupy at most $\min\{r,q(U)\}$ remaining slots, and ordering their groups first gives the best slots. Thus, $\Phi(U)$ is their maximum total click-through rate. For a payment vector $y=(y_i)_{i\in T}$, write $y(U)=\sum_{i\in U}y_i$.

The formula $y_i^\sigma(t)=t x_i^\sigma$ implies that every attainable expected payment vector has nonnegative coordinates. In addition, owners in $U$ receive at most $\Phi(U)$ total click-through rate at price $t$, so $y(U)\leq t\Phi(U)$. Items in $\mathcal G_{=}(F)$ fill all remaining $r$ slots, giving total payment $t\Phi(T)$. This defines the capacity region
\[
\mathcal P_F(t)
=
\left\{
y\in\R_{\geq0}^{T}:
\begin{array}{l}
y(U)\leq t\Phi(U)\quad(U\subseteq T),\\
y(T)=t\Phi(T)
\end{array}
\right\}.
\]

The subset inequalities express the best exposure that any set of advertisers can receive, so every attainable expected payment vector lies in $\mathcal P_F(t)$. To replace distributions over threshold-group orders by these capacity conditions in the equilibrium system, we also need the reverse inclusion. Otherwise, the capacity system could admit a payment vector that no ranking distribution realizes. The next lemma proves $\mathcal Y_F(t)=\mathcal P_F(t)$, ensuring both that every equilibrium remains represented and that any feasible payment vector later recovered from the capacity conditions can be realized by a distribution over threshold-group orders.

\begin{lemma}[Capacity Description of Attainable Payments]
\label{lem:pos-payment-polytope}
For every fixed $\theta\in F$, $\mathcal Y_F(t)=\mathcal P_F(t)$.
\end{lemma}

\begin{proof}
\emph{$(\Rightarrow)$ Necessity: $\mathcal Y_F(t)\subseteq\mathcal P_F(t)$.}
In any group order, $U$ receives at most $\Phi(U)$ total click-through rate and pays at most $t\Phi(U)$. Payments are nonnegative and total $t\Phi(T)$, so every $y^\sigma(t)$ lies in $\mathcal P_F(t)$. Probability-weighted averages preserve these conditions, proving inclusion.

\emph{$(\Leftarrow)$ Sufficiency: $\mathcal P_F(t)\subseteq\mathcal Y_F(t)$.}
For any linear objective $\sum_{i\in T}\omega_i y_i$, order advertisers so that $\omega_{i_1}\geq\cdots\geq\omega_{i_e}$, and let $U_k=\{i_1,\ldots,i_k\}$. This group order gives $\widehat y\in\mathcal Y_F(t)$ with
\[
\widehat y(U_k)=t\Phi(U_k)
\qquad(k=1,\ldots,e),
\]
because every group prefix takes the earliest available slots.

For any $y\in\mathcal P_F(t)$,
\[
\sum_{i\in T}\omega_i y_i
=
\omega_{i_e}y(T)
+
\sum_{k=1}^{e-1}
(\omega_{i_k}-\omega_{i_{k+1}})y(U_k).
\]

The coefficient differences are nonnegative. Using $y(T)=t\Phi(T)$ and $y(U_k)\leq t\Phi(U_k)$ gives
\[
\sum_{i\in T}\omega_i y_i
\leq
\sum_{i\in T}\omega_i\widehat y_i.
\]

Thus, each linear objective over $\mathcal P_F(t)$ is maximized by an actual group order's payment vector. A point outside the compact convex set $\mathcal Y_F(t)$ would admit a separating objective, contradicting this bound. Hence, $\mathcal P_F(t)\subseteq\mathcal Y_F(t)$, proving equality.
\end{proof}

\Cref{ex:threshold-capacity} illustrates how this description combines with maximal pacing to determine an equilibrium distribution.

\begin{example}[From Threshold Capacity to an Exact Equilibrium]
\label{ex:threshold-capacity}
Let $K=2$, $(c_1,c_2)=(1,1/2)$, and $T=\{1,2,3\}$, with group sizes $(q_1,q_2,q_3)=(2,1,1)$. Every product has value $2$, every advertiser has budget $1/2$, and set all multipliers to $1/2$. Thus all bids equal $t=1$, $h=0$, and $e=3>r=2$. In every group order, a whole other-owner group remains undisplayed, so each displayed product pays $1$ per click.

The capacity polytope is
\[
\mathcal P_F(1)
=
\{y\in\R_{\geq0}^3:
y_1+y_2+y_3=3/2,\ y_2\leq1,\ y_3\leq1\}.
\]
Advertiser~1 can fill both slots, whereas each other advertiser can receive at most the first slot's CTR. All remaining subset capacities follow from nonnegativity and the total-payment equality. Since every multiplier is below one, maximal pacing requires $y=(1/2,1/2,1/2)$, which lies in this polytope. It is realized by choosing the group orders $(2,1,3)$ and $(3,1,2)$ with probability $1/2$ each:
\[
\tfrac12(1/2,1,0)+\tfrac12(1/2,0,1)
=
(1/2,1/2,1/2).
\]
After expanding group~1 into its two products, these are complete valid rankings and all budgets bind. Hence the multipliers and this two-ranking distribution form an exact ABE-B. The capacity description provides the feasible expectation, and group orders realize it.
\end{example}

For a general threshold layer, we next intersect the capacity polytope with the remaining budget intervals, reduce its subset constraints, and recover a short distribution.

\noindentparagraph{\textbf{\textup{Remaining Budget Requirements.}}}
We now impose budgets. A payment variable $y_i$ for each advertiser in $T$ would add $e$ variables, unbounded in $K$. Instead, express budget requirements as intervals and test whether they contain a point of $\mathcal P_F(t)$, without keeping payment coordinates as unknowns.

For $i\in T$, let $y_i$ be its expected payment from $\mathcal G_{=}(F)$. Total payment is $p_i^++y_i$. Budget feasibility gives $y_i\leq B_i-p_i^+$, with equality required by maximal pacing for $i\in\mathsf P(F)$. Define
\[
\overline y_i=B_i-p_i^+,
\qquad
\underline y_i=
\begin{cases}
\overline y_i,&i\in\mathsf P(F),\\
0,&i\notin\mathsf P(F).
\end{cases}
\]

Advertisers with multiplier below one must pay exactly their remaining budget through $\mathcal G_{=}(F)$. Those with multiplier one may pay any nonnegative amount up to that budget. We must find
\[
y\in\mathcal P_F(t)
\quad\text{such that}\quad
\underline y_i\leq y_i\leq\overline y_i
\quad(i\in T).
\]

For $U\subseteq T$, write $\underline y(U)=\sum_{i\in U}\underline y_i$ and $\overline y(U)=\sum_{i\in U}\overline y_i$.

\noindentparagraph{\textbf{\textup{Eliminating Payment Variables.}}}
The cumulative CTR of the first $\min\{r,q(U)\}$ remaining slots is a nondecreasing concave function of the additive quantity $q(U)$. Hence $\Phi$ is submodular. The standard box-intersection criterion for submodular base polytopes~\cite{Sch03} gives the following necessary and sufficient conditions for a point in $\mathcal P_F(t)$ to satisfy the budget intervals:
\[
\begin{aligned}
&0\leq\underline y_i\leq\overline y_i
&& (i\in T),\\
&\underline y(U)\leq t\Phi(U)
&& (U\subseteq T),\\
&\overline y(T\setminus U)
\geq t\bigl(\Phi(T)-\Phi(U)\bigr)
&& (U\subseteq T).
\end{aligned}
\]

The second condition bounds each subset's minimum payment by its capacity. Since $U$ pays at most $t\Phi(U)$, the third requires enough budget outside $U$ to cover $t(\Phi(T)-\Phi(U))$. Together with valid coordinate bounds, these conditions are necessary and sufficient.

Substitute $\underline y$ and $\overline y$ into these inequalities. All conditions then depend only on $\theta$ and $\beta$. The bounds $\underline y_i$ and $\overline y_i$ abbreviate these expressions, so payment coordinates $y_i$ are no longer unknowns.

\noindentparagraph{\textbf{\textup{Reducing Subset Constraints.}}}
For $U\subseteq T$ with $q(U)<r$, its items cannot fill all remaining slots, so we include its subset conditions. If $q(U)\geq r$, then $\Phi(U)=\Phi(T)$. The corresponding lower bound follows from $\underline y(U)\leq\underline y(T)\leq t\Phi(T)$, and the complementary upper bound follows from nonnegativity because its right-hand side is zero.

Accordingly, define $\mathcal U=\{U\subseteq T:q(U)<r\}$.

Keep the total lower bound and test other subset inequalities only for $U\in\mathcal U$. Since $q_i\geq1$, each such subset has at most $r-1\leq K-1$ advertisers. Enumerate subsets of size at most $r-1$ and keep those with $q(U)<r$, giving
\[
|\mathcal U|
\leq
\sum_{k=0}^{r-1}\binom{e}{k}
=
O(n^{K-1})
\]
for fixed $K$. Include the empty set, whose complementary bound gives $\overline y(T)\geq t\Phi(T)$.

The reduced capacity description is also constructive: any feasible threshold-payment vector can be converted into a short, explicit distribution over threshold-group orders without enumerating all $e!$ possibilities.

The following proposition formalizes this decomposition and shows that at most $|T|$ threshold-group orders suffice, with polynomially many exact arithmetic operations.

\begin{proposition}[Constructive Decomposition of Threshold Payments]
\label{prop:threshold-decomposition}
Fix $t>0$. Given any $y\in\mathcal P_F(t)$, one can compute group orders $\sigma_1,\ldots,\sigma_m\in\Sigma_{=}(F)$ and nonnegative weights $\lambda_1,\ldots,\lambda_m$ such that
\[
y=\sum_{j=1}^{m}\lambda_j y^{\sigma_j}(t),
\qquad
\sum_{j=1}^{m}\lambda_j=1,
\qquad
m\leq |T|.
\]
The computation uses a polynomial number of exact arithmetic operations in the explicit description of $\mathcal P_F(t)$. In particular, algebraic $t$ and $y$ yield algebraic weights in the standard exact representation.
\end{proposition}

\begin{proof}
Write $e=|T|$ and represent $\mathcal P_F(t)$ by its total-payment equality together with the nonnegativity inequalities and the reduced capacity inequalities indexed by $\mathcal U$
\[
a_j^\top z\leq b_j
\qquad(j\in[m_{\mathrm{cap}}]).
\]
As shown above, the omitted capacity inequalities are implied by nonnegativity and the total-payment equality, so this is the same bounded polytope. Its affine dimension is at most $e-1$.

We maintain a residual point $z\in\mathcal P_F(t)$ and a coefficient $\mu>0$ such that the original vector $y$ equals a convex combination of already recorded order vectors plus $\mu z$. Initially, $z=y$ and $\mu=1$. Let
\[
\mathcal A(z)=\{j:a_j^\top z=b_j\}
\]
be the active inequalities and set $\omega=\sum_{j\in\mathcal A(z)}a_j$. Sort the advertisers in nonincreasing order of $\omega_i$ and let $v=y^\sigma(t)$ be the payment vector of the resulting group order. By the greedy argument in \Cref{lem:pos-payment-polytope}, $v$ maximizes $\omega^\top x$ over $\mathcal P_F(t)$. The point $z$ is also a maximizer, since every feasible $x$ satisfies
\[
\omega^\top x
=
\sum_{j\in\mathcal A(z)}a_j^\top x
\leq
\sum_{j\in\mathcal A(z)}b_j
=
\omega^\top z.
\]
Since $v$ is feasible and maximizes the same objective, equality holds throughout. Each active inequality has nonnegative slack at $v$, so equality of their sum implies $a_j^\top v=b_j$ for every $j\in\mathcal A(z)$. Hence $v$ lies in the minimal face of $\mathcal P_F(t)$ containing $z$.

If $z=v$, record $v$ with weight $\mu$ and terminate. Otherwise move from $z$ directly away from $v$. Define
\[
\eta
=
\min_{j:a_j^\top(z-v)>0}
\frac{b_j-a_j^\top z}{a_j^\top(z-v)}.
\]
The minimizing set is nonempty: if it were empty, the ray $z+\lambda(z-v)$ would remain feasible for every $\lambda\geq0$, contradicting boundedness. Every active inequality has denominator zero because it is also tight at $v$, while every inactive inequality has positive slack. Consequently $\eta>0$, and $z'=z+\eta(z-v)$ is feasible, preserves all previously active inequalities, and makes at least one additional inequality tight. Thus $z'$ lies in a proper face of the minimal face containing $z$. Finally,
\[
z=\frac{\eta}{1+\eta}v+\frac{1}{1+\eta}z'.
\]
Record $v$ with weight $\mu\eta/(1+\eta)$, replace $(z,\mu)$ by $(z',\mu/(1+\eta))$, and repeat.

Every nonterminal iteration strictly lowers the dimension of the minimal face containing the residual point. Since the affine dimension is at most $e-1$, the process records at most $e$ order vectors. The maintained convex-combination identity proves the claimed decomposition. Each iteration requires one sort, evaluation of the explicit inequalities, and exact arithmetic. Hence the overall procedure is polynomial-time.
\end{proof}

The proposition supplies at most $e$ threshold orders after a feasible expected payment vector is recovered. The nonlinear candidate system itself still contains no threshold-payment coordinates.

\noindentparagraph{\textbf{\textup{Candidate Equilibrium System.}}}
For each enumerated support $J$ from $\Sigma_{>}(F)$, impose the reduced conditions on $T$. Advertisers in $S\setminus T$ pay nothing from $\mathcal G_{=}(F)$, so their total payment is $p_i^+$. Those outside $S\cup T$ have multiplier one and payment zero. The system is
\begin{equation}
\label{eq:pos-many-system}
\left\{
\begin{aligned}
&\theta\in F,\quad \beta_\rho\geq0\quad(\rho\in J),\\
&\sum_{\rho\in J}\beta_\rho=1,\\
&0\leq\underline y_i\leq\overline y_i
&& (i\in T),\\
&\underline y(T)\leq t\Phi(T),\\
&\underline y(U)\leq t\Phi(U)
&& (U\in\mathcal U),\\
&\overline y(T\setminus U)
\geq t\bigl(\Phi(T)-\Phi(U)\bigr)
&& (U\in\mathcal U),\\
&p_i^+\leq B_i
&& (i\in S\setminus T),\\
&p_i^+=B_i
&& (i\in(S\setminus T)\cap\mathsf P(F)).
\end{aligned}
\right.
\end{equation}

Eliminating the threshold-payment coordinates leaves a system whose dimension depends only on $K$.

\begin{lemma}[Polynomial-Time Candidate Feasibility]
\label{lem:pos-many-feasibility}
For fixed $K$, the system in \Cref{eq:pos-many-system} has at most $2K$ variables and degree at most two. We can decide its feasibility and compute an exact solution whenever one exists, in polynomial time.
\end{lemma}

\begin{proof}[Proof of \Cref{lem:pos-many-feasibility}]
Only $\theta$ and the probabilities are unknown, giving $(s+1)+|J|\leq2(s+1)\leq2K$ variables. The expressions $p_i^+$, $\underline y_i$, and $\overline y_i$ add none. Each $p_i^{+,\rho}(\theta)$ is affine in $\theta$, and the paced set is fixed, so residual substitution introduces only products of $\beta_\rho$ with parameters in $\theta$. Thus, degree is at most two, while $t\Phi(U)$ and the other conditions are linear.

There are $O(M^2+n^{K-1})$ conditions with rational coefficients from values, CTRs, and recovery ratios. Apply the exact feasibility bound of \Cref{sec:framework}.
\end{proof}

\noindentparagraph{\textbf{\textup{Constructing the Equilibrium.}}}
Fix an admissible $(S,I,F)$ with $e>r$ and algebraic solution $(\theta^*,\beta^*)$, where $\theta^*=(t^*,\alpha_S^*)$. Define the full profile $\alpha^*$ by \Cref{eq:pos-recovery}, evaluate the budget intervals, and choose $y$ satisfying
\[
\left\{
\begin{aligned}
&\underline y_i^*\leq y_i\leq\overline y_i^*
&& (i\in T),\\
&y(U)\leq t^*\Phi(U)
&& (U\in\mathcal U),\\
&y(T)=t^*\Phi(T).
\end{aligned}
\right.
\]

The system in \Cref{eq:pos-many-system} ensures feasibility. Nonnegativity and the total-payment equality imply all omitted capacities. Thus, $y\in\mathcal P_F(t^*)$. With the equilibrium parameters fixed, this is a linear feasibility problem with $e$ variables and $O(e+|\mathcal U|)$ constraints.

\noindentparagraph{\textbf{\textup{Decomposing Payments into Group Orders.}}}
Apply \Cref{prop:threshold-decomposition} to $y\in\mathcal P_F(t^*)$. It returns at most $e$ threshold-group orders and exact weights
\[
y=\sum_{\sigma}\lambda_\sigma y^\sigma(t^*),
\qquad
\lambda_\sigma\geq0,
\qquad
\sum_\sigma\lambda_\sigma=1.
\]
The proposition is applied to the capacity polytope alone: the budget intervals constrain the recovered expectation $y$, not the individual threshold orders in its decomposition.

\noindentparagraph{\textbf{\textup{Constructing and Compressing Complete Rankings.}}}
For each $\rho\in J$ and recorded $\sigma$, concatenate them with the fixed valid $\mathcal G_{<}(F)$ order, expand groups, and assign probability $\beta_\rho^*\lambda_\sigma$. This preserves upper payments and $y$, so the residual bounds and \Cref{eq:pos-many-system} ensure an exact ABE-B at $\alpha^*$.

There are at most $(s+1)e\leq Ke$ rankings. Construct the support of size at most $n+1$ promised by \Cref{lem:outcome-compression} by eliminating linear dependencies among $(P(a,\alpha^*),1)$, where $P(a,\alpha^*)=(P_i(a,\alpha^*))_{i\in[n]}$ is the payment vector of ranking $a$. Merge duplicates and delete zero weights to obtain $\pi$, preserving multipliers, validity, and expected payments.

Since $|\mathcal U|=O(n^{K-1})$, payment recovery, the decomposition in \Cref{prop:threshold-decomposition}, ranking expansion, and final compression all take polynomial time for fixed $K$.

\noindentparagraph{\textbf{\textup{Correctness and Running Time.}}}
Let $N_{\mathrm{triple}}$ count admissible triples $(S,I,F)$.

\begin{proof}[Proof of \Cref{thm:pos-computation}]
For $e\leq r$, each algebraic solution of \Cref{eq:pos-few-system} gives valid rankings and expected payments satisfying budget feasibility and maximal pacing. For $e>r$, capacity and box-intersection conditions give attainable payments within residual bounds. Reconstruction yields an exact ABE-B and preserves upper payments at the solution. Final compression preserves all expected payments.

For completeness, take a group-consecutive ABE-B from \Cref{prop:exist}. \Cref{lem:pos-recovery} places its parameters in an enumerated admissible $(S,I,F)$. If $e\leq r$, a fixed lower-tail order preserves payments, and payment-space compression gives an enumerated support of at most $a_F+1$ rankings. If $e>r$, upper-payment compression gives at most $s+1$ orders, while $\mathcal G_{=}(F)$ payments satisfy the reduced conditions. Thus, an enumerated candidate system is feasible.

We next bound the running time. Each advertiser owns an item, so $n\leq M$.

\emph{Advertiser Sets and Threshold Intervals.} The number of candidate sets $S$ is
\[
\sum_{s=0}^{K-1}\binom ns
=
O(n^{K-1}).
\]

Splitting at distinct $u_i$ gives $O(n)$ intervals and singleton endpoints, hence $O(n^K)$ pairs $(S,I)$.

\emph{Comparison Cells.} For each $(S,I)$, $\theta=(t,\alpha_S)$ has dimension $s+1\leq K$ and $O(M^2+n)=O(M^2)$ comparison hyperplanes. Including equality faces gives $O(M^{2K})$ cells per pair, enumerated in $O(M^{O(K)})$ operations. Thus,
\[
N_{\mathrm{triple}}
=
O(n^K M^{2K})
=
O(M^{3K}).
\]

Testing threshold admissibility takes polynomial time per cell.

\emph{Representative Rankings and Supports.} For $e\leq r$, at most $h!e!\leq K!$ candidates yield supports of size at most $a_F+1\leq K+1$. For $e>r$, at most $h!\leq(K-1)!$ upper orders yield supports of size at most $s+1\leq K$. Hence, each triple has $O(1)$ supports for fixed $K$. Rankings and affine payments take $O(M^2)$ operations to construct.

\emph{Candidate Systems.} For $e\leq r$, systems have at most $2K+1$ variables and $O(M^2)$ conditions. For $e>r$, at most $2K$ variables remain. Only $\mathcal U$ requires explicit capacity tests, with
\[
|\mathcal U|
\leq
\sum_{\ell=0}^{r-1}\binom e\ell
=
O(n^{r-1})
=
O(n^{K-1}).
\]

Including cell and budget constraints, both cases have condition count
\[
m_F
=
O(M^2+n+n^{K-1})
=
O\bigl(M^{\max\{2,K-1\}}\bigr).
\]

All conditions have degree at most two. Applying the feasibility bound $(m_F+1)^{O(q^2)}$ with $q\leq2K+1$ gives $M^{O(K^3)}$ time per system.

\emph{Equilibrium Construction.} The solution weights give the distribution for $e\leq r$. For $e>r$, the reconstruction above recovers it in polynomially many operations in $M+n+|\mathcal U|$, hence $O(M^{O(K)})$. Combining distributions and compressing payments gives at most $n+1$ rankings within this bound.

There are $O(M^{3K})$ admissible triples and constantly many supports per triple for fixed $K$. Together with equilibrium construction, the total running time is $O(M^{O(K^3)})$.
\end{proof}

\subsection{ROI Constraints: Sparse Nonunit Multipliers}
\label{sec:roi-positions}

An ROI advertiser that is undisplayed at a given multiplier profile has zero value and payment at that profile, so the display threshold alone does not determine its multiplier. We obtain a small set of unknown multipliers by preserving the displayed items and their price witnesses while lowering the remaining multipliers to one.

\begin{theorem}[Exact ABE-R Computation with Few Slots]
\label{thm:roi-pos-computation}
For every fixed number of display slots $K$, there exists a polynomial-time algorithm that computes an exact ABE-R $(\alpha,\pi)$ with time complexity $O(M^{O(K^2)})$.
\end{theorem}

\begin{algorithm}[H]
\small
\caption{Exact ABE-R with a Fixed Number of Display Slots}
\label{alg:roi-positions}
\begin{algorithmic}[1]
\REQUIRE An explicit rational ROI instance with fixed $K$.
\ENSURE An exact ABE-R with at most $K+2$ nonunit multipliers and support size at most $\min\{n+1,K+3\}$.
\STATE \textbf{Phase I: Sparse-Multiplier Search.}
\STATE Compute value groups and the cap $A$.
\FORALL{$S\subseteq[n]$ with $s=|S|\leq\min\{n,K+2\}$}
    \STATE Set $\alpha_i=1$ outside $S$ and enumerate comparison cells $F$ in $[1,A]^s$, including ties and multiplier endpoints.
    \FORALL{cells $F$}
        \STATE \textbf{Phase II: Candidate Ranking Construction.}
        \STATE Enumerate admissible displayed group prefixes and complete their canonical tails.
        \STATE Retain candidates with zero surplus outside $S$ throughout $F$.
        \STATE Compute their affine surplus functions on $S$.
        \STATE Enumerate nonempty candidate supports $J$ with $|J|\leq s+1$.
        \STATE \textbf{Phase III: Equilibrium Construction.}
        \FORALL{candidate supports produced in Phase II}
            \STATE Form \Cref{eq:roi-pos-system} for $F$ and $J$.
            \STATE Apply exact semialgebraic sampling.
            \IF{an algebraic solution $(\alpha_S^*,w^*)$ is found}
                \STATE Recover the full multiplier profile and expand the selected group orders into complete rankings.
                \STATE Assign the sampled weights, merge identical rankings, and remove zero weights.
                \STATE \textbf{return} the resulting exact equilibrium $(\alpha^*,\pi)$.
            \ENDIF
        \ENDFOR
    \ENDFOR
\ENDFOR
\end{algorithmic}
\end{algorithm}

The algorithm returns a profile with at most $K+2$ multipliers different from one and a distribution $\pi$ supported on at most $\min\{n+1,K+3\}$ group-consecutive rankings. The structural task is therefore to prove that such a sparse representative always exists. Algorithm~\ref{alg:roi-positions} summarizes the construction in three phases: sparse-multiplier search, candidate ranking construction, and equilibrium construction.

We observe that the number of advertisers with an above-threshold item is bounded by the number of display slots. For a group-consecutive ABE-R, let $t$ be the $K$th highest item bid and $h$ the item count strictly above $t$. Let $T$ be owners with a group bidding $t$, set $e=|T|$ and $r=K-h$, and define $H_t=\{i:\alpha_i u_i>t\}$. Each advertiser in $H_t$ owns an above-threshold item, so $|H_t|\leq h<K$.

For a displayed owner, the relevant tail price is the largest bid after excluding that owner. The two highest bids from distinct owners suffice for all such exclusions. If a nonunit threshold owner is never displayed, its unchanged bid can instead witness the threshold price for every displayed owner. The next two lemmas formalize these observations.

\begin{lemma}[Two Below-Threshold Price Witnesses]
\label{lem:two-tail-witnesses}
Fix a multiplier profile and a threshold $t>0$. For each owner $i$, let $d_i=\max\{\alpha_i v_{ij}:\alpha_i v_{ij}<t\}$, with $\max\varnothing=0$. There is a set $\mathcal W$ of at most two owners such that, for every owner $i$,
\[
\max_{j\neq i}d_j
=
\max_{j\in \mathcal W\setminus\{i\}}d_j,
\]
where the maximum of an empty set is zero. Keeping the bids of owners in $\mathcal W$ and lowering only bids that were already below $t$ preserves this eligible maximum for every possible displayed owner.
\end{lemma}

\begin{proof}
If $d_i=0$ for every owner $i$, take $\mathcal W=\varnothing$, and both maxima are zero. Otherwise, choose an owner $i_1$ with the largest positive $d_i$ and set $\mathcal W=\{i_1\}$. If another owner has positive $d_i$, choose an owner $i_2\neq i_1$ with the second-largest positive $d_i$ and add it to $\mathcal W$. For owner $i\neq i_1$, the eligible maximum is attained by $i_1$. For $i=i_1$, it is attained by $i_2$ if selected, and is zero otherwise. The final statement follows because a preserved witness attains the original maximum and no modified bid can exceed it.
\end{proof}

When more threshold groups exist than remaining slots, one threshold owner can play an even stronger role: it supplies the same bid-$t$ witness to every displayed item.

\begin{lemma}[Universal Threshold Witness]
\label{lem:universal-threshold-witness}
Consider a group-consecutive exact ABE-R with threshold data $(t,h,r,T,e,H_t)$ as defined above, and suppose $e>r$. Let $\mathcal N=\{i\in T\setminus H_t:\alpha_i>1\}$. Every owner in $\mathcal N$ is undisplayed in every ranking in the equilibrium support. For any $i^\dagger\in\mathcal N$, its threshold group can therefore remain below the displayed prefix and supply an eligible bid-$t$ price witness for every displayed item. If $\mathcal N=\varnothing$, all threshold owners outside $H_t$ already have multiplier one.
\end{lemma}

\begin{proof}
Fix $i\in \mathcal N$, so $\alpha_i>1$. Membership in $T\setminus H_t$ implies $\alpha_i u_i=t$. Since $e>r$, at most $r$ threshold groups can intersect the first $K$ positions, and at least one whole threshold group of another owner remains below every displayed threshold item. All displayed items of $i$ therefore belong to its top-value threshold group and pay $t$ per click. If their expected total click-through rate were $x>0$, then
\[
V_i(\pi)=u_i x,
\qquad
P_i(\pi,\alpha)=t x=\alpha_i u_i x=\alpha_i V_i(\pi)>V_i(\pi),
\]
contradicting ROI feasibility. Hence $i$ is undisplayed in every supported ranking. An owner $i^\dagger\in\mathcal N$ is different from every displayed owner and its threshold group can be placed immediately below the displayed prefix, so it supplies an eligible bid-$t$ witness. The final statement is immediate from the definition of $\mathcal N$.
\end{proof}

The witness lemmas give an outcome-preserving sparsity result.

\begin{lemma}[Sparse Nonunit Multipliers]
\label{lem:roi-pos-sparse}
Every group-consecutive exact ABE-R can be transformed into a group-consecutive exact ABE-R with at most $K+2$ multipliers different from one by lowering selected multipliers to one and rearranging undisplayed groups. The transformation preserves the item, value, and price in every displayed slot of every supported ranking. Consequently, every ROI instance admits such a sparse equilibrium.
\end{lemma}

\begin{proof}
Start from any group-consecutive exact ABE-R $(\alpha,\pi)$, whose existence is guaranteed by \Cref{prop:exist}. All threshold objects below refer to this profile. We construct $\alpha'$ and transform each supported ranking.

\smallskip
\noindent\emph{Case 1: $e\leq r$.}
Let $\mathcal W$ be the set from \Cref{lem:two-tail-witnesses} and define
\[
\mathcal R=H_t\cup T\cup \mathcal W,
\qquad
\alpha'_i=
\begin{cases}
\alpha_i,&i\in\mathcal R,\\
1,&i\notin\mathcal R.
\end{cases}
\]
Every owner outside $H_t\cup T$ has all bids strictly below $t$, and resetting its multiplier to one can only lower those bids. For a supported ranking $a$, keep the relative order of all groups whose original bid is at least $t$, then append the remaining groups in nonincreasing order of their new bids. The resulting group-consecutive ranking $a'$ is valid at $\alpha'$ and has the same first $K$ items as $a$.

Fix a displayed item owned by $i$. If its original price is set by an other-owner group in the preserved at-least-$t$ block, that group and bid are unchanged and all appended bids are below $t$. Otherwise its price is the highest eligible below-threshold bid, namely $\max_{j\neq i}d_j$, which is preserved by \Cref{lem:two-tail-witnesses}. Hence every displayed value and price is unchanged. The number of possibly nonunit multipliers is at most
\[
|H_t|+|T|+2
\leq h+e+2
\leq h+r+2
=K+2.
\]

\smallskip
\noindent\emph{Case 2: $e>r$.}
Use \Cref{lem:universal-threshold-witness}. If $\mathcal N\neq\varnothing$, choose one owner $i^\dagger\in\mathcal N$. Otherwise, introduce no additional witness. Set
\[
\mathcal R=H_t\cup\{i^\dagger\},
\qquad
\alpha'_i=
\begin{cases}
\alpha_i,&i\in\mathcal R,\\
1,&i\notin\mathcal R,
\end{cases}
\]
omitting $i^\dagger$ when $\mathcal N=\varnothing$. Any changed owner outside $H_t$ either belongs to $\mathcal N$ and is never displayed, or originally has all bids strictly below $t$ and remains below $t$ after the reset.

For a supported ranking $a$, preserve the relative order of all above-threshold groups and of the threshold groups intersecting the first $K$ positions, expanding each displayed group completely. These displayed threshold groups still bid $t$: their owners lie in $H_t$ or have multiplier one. If $\mathcal N\neq\varnothing$, place the group of $i^\dagger$ next and append all remaining groups in nonincreasing new-bid order, including any other groups still bidding $t$. If $\mathcal N=\varnothing$, preserve the entire remaining threshold layer and sort the lower groups by their new bids. The final statement of \Cref{lem:universal-threshold-witness} ensures that this threshold layer is unchanged. The resulting complete ranking $a'$ is valid and group-consecutive, with the same first $K$ items as $a$.

A displayed price strictly above $t$ is determined inside the preserved above-threshold block and is unchanged. Every other displayed price equals $t$: a threshold group remains below the displayed prefix and has an owner different from the displayed owner. When $\mathcal N\neq\varnothing$, the owner $i^\dagger$ supplies this witness universally. When $\mathcal N=\varnothing$, the preserved threshold layer supplies the same witnesses as before. Thus values and prices are again preserved slot by slot, with at most $|H_t|+1\leq h+1\leq K$ nonunit multipliers.

In either case, let $\pi'$ aggregate the probabilities of rankings mapped to the same $a'$. Advertisers whose multipliers are unchanged have the same expected values and payments as before. Every advertiser whose multiplier is lowered to one is undisplayed in all transformed rankings and therefore has value and payment zero. Since $1<A$, its required surplus equality holds. Hence $(\alpha',\pi')$ is a group-consecutive exact ABE-R with at most $K+2$ nonunit multipliers and the claimed preservation of displayed outcomes.
\end{proof}

\noindentparagraph{\textbf{\textup{From Sparsity to a Search Space.}}}
Enumerate $S\subseteq[n]$ with $s=|S|\leq\min\{n,K+2\}$. The set $S$ need only contain all nonunit advertisers, and an advertiser in $S$ may still receive multiplier one. Leave $\alpha_S\in[1,A]^s$ unknown and set $\alpha_i=1$ outside $S$. Thus, a group bid is $v_g\alpha_{i(g)}$ when its owner lies in $S$, and the constant $v_g$ otherwise. Compare cross-owner group bids and include the endpoint forms $\alpha_i-1$ and $\alpha_i-A$ for $i\in S$. Handle constant comparisons as in the fixed-slot budget construction and include equality faces.

For each cell $F$, all group-bid comparisons and the set $\mathsf C_S(F)=\{i\in S:\alpha_i<A\text{ throughout }F\}$ are fixed. \Cref{lem:roi-pos-sparse} guarantees that some enumerated $S$ and cell contain an exact equilibrium.

\noindentparagraph{\textbf{\textup{Candidate Rankings and Payment Representation.}}}
A group-consecutive ranking reaches slot $K$ within at most $K$ groups. Enumerate distinct-group lists $(g_1,\ldots,g_p)$ with $p\leq K$, nonincreasing bids on $F$, and
\[
\sum_{j<p}q_{g_j}<K
\leq
\sum_{j\leq p}q_{g_j}.
\]

A list is admissible if every omitted group bids at most $b_{g_p}$. Expand its groups in their prescribed internal orders, including the undisplayed remainder of $g_p$ when the boundary falls inside that group. Append the omitted groups in decreasing bid order, resolving ties by a fixed group order. This gives a representative ranking valid throughout $F$.

This gives $O(M^K)$ complete candidates covering every valid group-consecutive displayed prefix. Canonical tail completion preserves displayed values and payments: all tail items follow the displayed items, and their highest eligible other-owner bid is independent of tail order. Tail items remain eligible for pricing when candidate payments are computed. For a fixed complete candidate, each pricing representative is fixed, making values constant and payments affine in $\alpha_S$.

\noindentparagraph{\textbf{\textup{Zero-Surplus Rankings and Compression.}}}
Outside $S$, every multiplier is one, so a displayed item's price is at most its value. Each such advertiser therefore has nonnegative surplus in every valid ranking. Since $1<A$, equilibrium requires zero expected surplus, and each ranking with positive probability must give zero surplus outside $S$.

This condition can be determined throughout the cell. Because displayed CTRs are positive, every displayed item of an advertiser outside $S$ must have its price equal to its bid. The cell fixes whether its other-owner pricing representative has that bid. If no representative exists, the price is zero and the condition fails. Undisplayed advertisers have value and payment zero. Thus, we can restrict the representative family to rankings satisfying the zero-surplus condition before choosing numerical multipliers.

Let $L_F$ be the number of retained rankings. Index them by $\ell\in[L_F]$ and write their affine surplus as $R_i^{(\ell)}(\alpha_S)=V_i(a^{(\ell)})-P_i(a^{(\ell)},\alpha)$. Their surplus vanishes outside $S$. Applying \Cref{lem:outcome-compression} to the $s$ coordinates on $S$ gives support at most $s+1$. For each nonempty $J\subseteq[L_F]$ with $|J|\leq s+1$, define
\[
r_i(\alpha_S,w)
=
\sum_{\ell\in J}w_\ell R_i^{(\ell)}(\alpha_S)
\qquad(i\in S).
\]

\noindentparagraph{\textbf{\textup{Equilibrium Constraints.}}}
For each selected cell and support, impose
\begin{equation}
\label{eq:roi-pos-system}
\left\{
\begin{aligned}
&\alpha_S\in F,\\
&w_\ell\geq0\quad(\ell\in J),\qquad
\sum_{\ell\in J}w_\ell=1,\\
&r_i(\alpha_S,w)\geq0
&& (i\in S),\\
&r_i(\alpha_S,w)=0
&& (i\in\mathsf C_S(F)).
\end{aligned}
\right.
\end{equation}

\noindentparagraph{\textbf{\textup{Recovering the Equilibrium.}}}
For an algebraic solution $(\alpha_S^*,w^*)$, set
\[
\alpha_i^*=
\begin{cases}
(\alpha_S^*)_i,&i\in S,\\
1,&i\notin S.
\end{cases}
\]

Expand the associated complete rankings, including their canonical tails. The zero-surplus condition gives
\[
R_i(\pi,\alpha^*)=
\begin{cases}
r_i(\alpha_S^*,w^*),&i\in S,\\
0,&i\notin S.
\end{cases}
\]

The candidate system enforces the ROI conditions on $S$, while outside $S$ the assigned multiplier is $1<A$ and surplus is zero. As in the fixed-slot ABE-B construction, recovery completes the multiplier profile and expands weighted group orders into complete rankings. Here, the sampled probabilities already determine the distribution, so no separate threshold-payment recovery or decomposition is needed. Recovery takes $O(n+(s+1)M)$ operations and leaves at most $s+1\leq\min\{n+1,K+3\}$ rankings.

\noindentparagraph{\textbf{\textup{Correctness and Running Time.}}}
\begin{proof}[Proof of \Cref{thm:roi-pos-computation}]
Each admissible representative is valid throughout its cell with zero surplus outside $S$. The system in \Cref{eq:roi-pos-system} enforces the remaining ROI conditions. Every algebraic solution thus gives an exact ABE-R.

For completeness, take the sparse equilibrium from \Cref{lem:roi-pos-sparse}. Its nonunit multipliers lie in an enumerated $S$ and its profile in an enumerated cell. Canonical prefix completions preserve values and payments and satisfy the zero-surplus condition outside $S$. Compressing surplus on $S$ gives an enumerated support of at most $s+1$ rankings. Thus, a candidate system is feasible.

There are $O(n^{K+2})$ choices of $S$, at most $O(M^{2(K+2)})$ cells each, and $O(M^K)$ candidates per cell. Supports of at most $K+3$ give at most
\[
O\left(n^{K+2}M^{2(K+2)+K(K+3)}\right)
=
O\left(M^{K^2+6K+6}\right)
\]
candidate systems, using $n\leq M$.

Each system has $s+|J|\leq2s+1\leq2K+5$ variables, degree at most two, and $O(M^2+K)$ conditions. The exact feasibility bound of \Cref{sec:framework} costs $O(M^{O(K^2)})$ operations per system and in total. Prefix construction, filtering, and output fit this bound.
\end{proof}

\begin{proof}[Proof of \Cref{thm:main}]
The first structural principle follows from \Cref{thm:adv-computation,thm:roi-computation}, and the second structural principle follows from \Cref{thm:pos-computation,thm:roi-pos-computation}. Every advertiser owns at least one item, so $n\le M$, and the fixed-item case follows from the fixed-advertiser algorithms. Each algorithm computes exact algebraic multipliers and an explicit distribution over valid rankings.
\end{proof}

\section{Conclusion}

Exact budget and ROI autobidding equilibria in GSP sponsored shopping can be computed in polynomial time with either few advertisers or few display slots. Layerwise convexification gives a polynomial representation of ranking-dependent outcomes when the advertiser dimension is fixed. At the display boundary, threshold recovery and capacity elimination handle budgets, while outcome-preserving multiplier sparsity handles ROI. Each construction returns an explicit distribution over complete valid rankings.

A natural next question fixes the number of products per advertiser while allowing both advertisers and slots to grow. Tie layers can then contain many owners, and the display boundary may involve many relevant multipliers. Resolving equilibrium computation in this regime would further explain how catalog size and ownership structure contribute to the complexity of one global ranking.

\section*{AI Usage Disclosure}

The algorithms and the core ideas underlying this paper were developed independently by the authors without AI assistance. The formulation of the research questions, the development of the algorithmic approaches, and the main mathematical insights originated from the authors' own work.

During the preparation and revision of the manuscript, we used GPT-6.1-SOL to assist with checking mathematical derivations, examining the consistency of definitions and notation, and identifying potential gaps or ambiguities in the presentation of proofs. The model was also used to assist with preparing and refining figures, improving the organization and clarity of the exposition, and drafting and revising the \LaTeX{} source.

The authors evaluated the model's suggestions and independently reviewed and verified the resulting mathematical arguments, figures, and text. The authors take full responsibility for all algorithms, results, proofs, and statements in the paper, as well as for the accuracy and integrity of the final manuscript.

\bibliographystyle{plainnat}
\bibliography{references}

@misc{HLN26,
  author       = {Hansen, Kristoffer Arnsfelt and Li, Yuhao and Nie, Xinhao},
  title        = {Complexity of Autobidding Equilibria in {GSP} Sponsored Shopping},
  year         = {2026},
  howpublished = {arXiv:2610.13716},
  note         = {Version 1},
  doi          = {10.48550/arXiv.2610.13716},
  url          = {https://arxiv.org/abs/2610.13716v1}
}

@misc{HYWL26,
  author       = {Huang, Yiyang and Yan, Yonglei and Wang, Zihe and Liu, Zhengyang},
  title        = {Pacing Equilibria in Second-Price Auctions with Few Goods},
  year         = {2026},
  howpublished = {arXiv:2605.09332},
  doi          = {10.48550/arXiv.2605.09332},
  url          = {https://arxiv.org/abs/2605.09332}
}

@misc{DLPST26,
  author       = {D{\"u}tting, Paul and Li, Yuhao and Paes Leme, Renato and Spendlove, Kelly and Teng, Yifeng},
  title        = {Autobidding Equilibria in Sponsored Shopping},
  year         = {2026},
  howpublished = {arXiv:2602.21966},
  doi          = {10.48550/arXiv.2602.21966},
  url          = {https://arxiv.org/abs/2602.21966v2}
}

@article{YWL26,
  author  = {Yan, Yonglei and Wang, Zihe and Liu, Zhengyang},
  title   = {Pacing Equilibria in Second-Price Auctions with Few Buyers},
  journal = {Proceedings of the AAAI Conference on Artificial Intelligence},
  year    = {2026},
  volume  = {40},
  number  = {20},
  pages   = {17302--17309},
  doi     = {10.1609/aaai.v40i20.38782}
}

@inproceedings{ABM19,
  author    = {Aggarwal, Gagan and Badanidiyuru, Ashwinkumar and Mehta, Aranyak},
  title     = {Autobidding with Constraints},
  booktitle = {Web and Internet Economics},
  editor    = {Caragiannis, Ioannis and Mirrokni, Vahab and Nikolova, Evdokia},
  series    = {Lecture Notes in Computer Science},
  volume    = {11920},
  year      = {2019},
  pages     = {17--30},
  publisher = {Springer},
  doi       = {10.1007/978-3-030-35389-6_2}
}

@article{CKSS22,
  author  = {Conitzer, Vincent and Kroer, Christian and Sodomka, Eric and Stier-Moses, Nicolas E.},
  title   = {Multiplicative Pacing Equilibria in Auction Markets},
  journal = {Operations Research},
  year    = {2022},
  volume  = {70},
  number  = {2},
  pages   = {963--989},
  doi     = {10.1287/opre.2021.2167}
}

@article{CKK24,
  author  = {Chen, Xi and Kroer, Christian and Kumar, Rachitesh},
  title   = {The Complexity of Pacing for Second-Price Auctions},
  journal = {Mathematics of Operations Research},
  year    = {2024},
  volume  = {49},
  number  = {4},
  pages   = {2109--2135},
  doi     = {10.1287/moor.2022.0009}
}

@article{CKPSSSW22,
  author  = {Conitzer, Vincent and Kroer, Christian and Panigrahi, Debmalya and Schrijvers, Okke and Stier-Moses, Nicolas E. and Sodomka, Eric and Wilkens, Christopher A.},
  title   = {Pacing Equilibrium in First Price Auction Markets},
  journal = {Management Science},
  year    = {2022},
  volume  = {68},
  number  = {12},
  pages   = {8515--8535},
  doi     = {10.1287/mnsc.2022.4310}
}

@article{EOS07,
  author  = {Edelman, Benjamin and Ostrovsky, Michael and Schwarz, Michael},
  title   = {Internet Advertising and the Generalized Second-Price Auction: Selling Billions of Dollars Worth of Keywords},
  journal = {American Economic Review},
  year    = {2007},
  volume  = {97},
  number  = {1},
  pages   = {242--259},
  doi     = {10.1257/aer.97.1.242}
}

@article{Varian07,
  author  = {Varian, Hal R.},
  title   = {Position Auctions},
  journal = {International Journal of Industrial Organization},
  year    = {2007},
  volume  = {25},
  number  = {6},
  pages   = {1163--1178},
  doi     = {10.1016/j.ijindorg.2006.10.002}
}

@book{BPR06,
  author    = {Basu, Saugata and Pollack, Richard and Roy, Marie-Fran{\c{c}}oise},
  title     = {Algorithms in Real Algebraic Geometry},
  edition   = {2},
  series    = {Algorithms and Computation in Mathematics},
  volume    = {10},
  year      = {2006},
  publisher = {Springer},
  doi       = {10.1007/3-540-33099-2}
}

@book{Sch03,
  author    = {Schrijver, Alexander},
  title     = {Combinatorial Optimization: Polyhedra and Efficiency},
  year      = {2003},
  publisher = {Springer}
}

@book{Ziegler95,
  author    = {Ziegler, G{\"u}nter M.},
  title     = {Lectures on Polytopes},
  series    = {Graduate Texts in Mathematics},
  volume    = {152},
  year      = {1995},
  publisher = {Springer},
  doi       = {10.1007/978-1-4613-8431-1}
}

@book{Edelsbrunner87,
  author    = {Edelsbrunner, Herbert},
  title     = {Algorithms in Combinatorial Geometry},
  series    = {Monographs in Theoretical Computer Science. An EATCS Series},
  volume    = {10},
  year      = {1987},
  publisher = {Springer},
  doi       = {10.1007/978-3-642-61568-9}
}

@inproceedings{LucierPSZ24,
  author    = {Lucier, Brendan and Pattathil, Sarath and Slivkins, Aleksandrs and Zhang, Mengxiao},
  title     = {Autobidders with Budget and {ROI} Constraints: Efficiency, Regret, and Pacing Dynamics},
  booktitle = {Proceedings of the Thirty-Seventh Conference on Learning Theory},
  series    = {Proceedings of Machine Learning Research},
  volume    = {247},
  pages     = {3642--3643},
  year      = {2024},
  publisher = {PMLR},
  url       = {https://proceedings.mlr.press/v247/lucier24a.html}
}

@article{MehtaSVV07,
  author  = {Mehta, Aranyak and Saberi, Amin and Vazirani, Umesh V. and Vazirani, Vijay V.},
  title   = {{AdWords} and Generalized Online Matching},
  journal = {Journal of the ACM},
  year    = {2007},
  volume  = {54},
  number  = {5},
  pages   = {22:1--22:19},
  doi     = {10.1145/1284320.1284321}
}

@inproceedings{FeldmanMPS07,
  author    = {Feldman, Jon and Muthukrishnan, S. and P{\'a}l, Martin and Stein, Cliff},
  title     = {Budget Optimization in Search-Based Advertising Auctions},
  booktitle = {Proceedings of the 8th ACM Conference on Electronic Commerce},
  year      = {2007},
  pages     = {40--49},
  publisher = {ACM},
  doi       = {10.1145/1250910.1250917}
}

@inproceedings{BorgsCEIJM07,
  author    = {Borgs, Christian and Chayes, Jennifer T. and Etesami, Omid and Immorlica, Nicole and Jain, Kamal and Mahdian, Mohammad},
  title     = {Dynamics of Bid Optimization in Online Advertisement Auctions},
  booktitle = {Proceedings of the 16th International Conference on World Wide Web},
  year      = {2007},
  pages     = {531--540},
  publisher = {ACM},
  doi       = {10.1145/1242572.1242644}
}

@inproceedings{EvenDarMMMN09,
  author    = {Even-Dar, Eyal and Mansour, Yishay and Mirrokni, Vahab S. and Muthukrishnan, S. and Nadav, Uri},
  title     = {Bid Optimization for Broad Match Ad Auctions},
  booktitle = {Proceedings of the 18th International Conference on World Wide Web},
  year      = {2009},
  pages     = {231--240},
  publisher = {ACM},
  doi       = {10.1145/1526709.1526741}
}

@inproceedings{DengMMZ21,
  author    = {Deng, Yuan and Mao, Jieming and Mirrokni, Vahab S. and Zuo, Song},
  title     = {Towards Efficient Auctions in an Auto-Bidding World},
  booktitle = {Proceedings of the Web Conference 2021},
  year      = {2021},
  pages     = {3965--3973},
  publisher = {ACM},
  doi       = {10.1145/3442381.3450052}
}

@inproceedings{BalseiroDMMZ21,
  author    = {Balseiro, Santiago R. and Deng, Yuan and Mao, Jieming and Mirrokni, Vahab S. and Zuo, Song},
  title     = {The Landscape of Auto-Bidding Auctions: Value versus Utility Maximization},
  booktitle = {Proceedings of the 22nd ACM Conference on Economics and Computation},
  year      = {2021},
  pages     = {132--133},
  publisher = {ACM},
  doi       = {10.1145/3465456.3467607}
}

@inproceedings{DengMMMZZ24,
  author    = {Deng, Yuan and Mahdian, Mohammad and Mao, Jieming and Mirrokni, Vahab S. and Zhang, Hanrui and Zuo, Song},
  title     = {Efficiency of the Generalized Second-Price Auction for Value Maximizers},
  booktitle = {Proceedings of the ACM Web Conference 2024},
  year      = {2024},
  pages     = {46--56},
  publisher = {ACM},
  doi       = {10.1145/3589334.3645360}
}

@inproceedings{DengGJLM24,
  author    = {Deng, Yuan and Golrezaei, Negin and Jaillet, Patrick and Liang, Jason Cheuk Nam and Mirrokni, Vahab S.},
  title     = {Individual Welfare Guarantees in the Autobidding World with Machine-Learned Advice},
  booktitle = {Proceedings of the ACM Web Conference 2024},
  year      = {2024},
  pages     = {267--275},
  publisher = {ACM},
  doi       = {10.1145/3589334.3645660}
}

@article{BalseiroKMM21,
  author  = {Balseiro, Santiago R. and Kim, Anthony and Mahdian, Mohammad and Mirrokni, Vahab S.},
  title   = {Budget-Management Strategies in Repeated Auctions},
  journal = {Operations Research},
  year    = {2021},
  volume  = {69},
  number  = {3},
  pages   = {859--876},
  doi     = {10.1287/opre.2020.2073}
}

@inproceedings{ZhuWYWYC09,
  author    = {Zhu, Yunzhang and Wang, Gang and Yang, Junli and Wang, Dakan and Yan, Jun and Chen, Zheng},
  title     = {Revenue Optimization with Relevance Constraint in Sponsored Search},
  booktitle = {Proceedings of the Third International Workshop on Data Mining and Audience Intelligence for Advertising},
  year      = {2009},
  pages     = {55--60},
  publisher = {ACM},
  doi       = {10.1145/1592748.1592756}
}

@inproceedings{ZhangWMMHW18,
  author    = {Zhang, Weiru and Wei, Chao and Meng, Xiaonan and Hu, Yi and Wang, Hao},
  title     = {The Whole-Page Optimization via Dynamic Ad Allocation},
  booktitle = {Companion Proceedings of the Web Conference 2018},
  year      = {2018},
  pages     = {1407--1411},
  publisher = {International World Wide Web Conferences Steering Committee},
  doi       = {10.1145/3184558.3191584}
}

@inproceedings{WuHHL18,
  author    = {Wu, Liang and Hu, Diane and Hong, Liangjie and Liu, Huan},
  title     = {Turning Clicks into Purchases: Revenue Optimization for Product Search in E-Commerce},
  booktitle = {Proceedings of the 41st International ACM SIGIR Conference on Research and Development in Information Retrieval},
  year      = {2018},
  pages     = {365--374},
  publisher = {ACM},
  doi       = {10.1145/3209978.3209993}
}

@inproceedings{GeZLI21,
  author    = {Ge, Zhen and Zhou, Wei and Lute, Jesse and Ilardi, Adam},
  title     = {Relevance Constrained Re-Ranking in Sponsored Listing Recommendations},
  booktitle = {Proceedings of the 27th ACM SIGKDD Workshop on Data Mining for Online Advertising},
  year      = {2021},
  pages     = {1--7},
  publisher = {ACM}
}

@inproceedings{DashGMCKG24,
  author    = {Dash, Abhisek and Ghosh, Saptarshi and Mukherjee, Animesh and Chakraborty, Abhijnan and Gummadi, Krishna P.},
  title     = {Sponsored Is the New Organic: Implications of Sponsored Results on Quality of Search Results in the Amazon Marketplace},
  booktitle = {Proceedings of the AAAI/ACM Conference on AI, Ethics, and Society},
  year      = {2024},
  volume    = {7},
  pages     = {316--327},
  doi       = {10.1609/aies.v7i1.31639}
}

@misc{LuZZ26,
  author       = {Lu, Haihao and Zhang, Luyang and Zhu, Yuting},
  title        = {The Power of Linear Programming in Sponsored Listings Ranking: Evidence from a Large-Scale Field Experiment},
  year         = {2026},
  howpublished = {arXiv:2403.14862},
  doi          = {10.48550/arXiv.2403.14862}
}

@inproceedings{BalseiroBF24,
  author    = {Balseiro, Santiago R. and Bhawalkar, Kshipra and Feng, Zhe and Lu, Haihao and Mirrokni, Vahab and Sivan, Balasubramanian and Wang, Di},
  title     = {A Field Guide for Pacing Budget and {ROS} Constraints},
  booktitle = {Proceedings of the 41st International Conference on Machine Learning},
  series    = {Proceedings of Machine Learning Research},
  volume    = {235},
  pages     = {2607--2638},
  year      = {2024},
  publisher = {PMLR},
  url       = {https://proceedings.mlr.press/v235/balseiro24a.html}
}

@inproceedings{LT24,
  author    = {Li, Juncheng and Tang, Pingzhong},
  title     = {Vulnerabilities of Single-Round Incentive Compatibility in Auto-bidding: Theory and Evidence from {ROI}-Constrained Online Advertising Markets},
  booktitle = {Proceedings of the Thirty-Third International Joint Conference on Artificial Intelligence},
  pages     = {2886--2894},
  year      = {2024},
  publisher = {International Joint Conferences on Artificial Intelligence Organization},
  doi       = {10.24963/ijcai.2024/320},
  url       = {https://www.ijcai.org/proceedings/2024/320}
}

\end{document}